\documentclass[superscriptaddress,
 amsmath,amssymb,
 aps,
pra,
reprint
]{revtex4-2}

\usepackage{graphicx}
\usepackage{dcolumn}
\usepackage{physics}
\usepackage{amsfonts, amsmath}
\usepackage{dsfont}
\usepackage{amsthm}
\usepackage{bm}
\usepackage{hyperref}
\usepackage{bbold}
\usepackage{color}
\usepackage{lipsum,babel}
\usepackage[normalem]{ulem}
\usepackage{array}
\newcolumntype{L}{>{\raggedright\arraybackslash}X}
\usepackage{tabularx}
\DeclareMathOperator*{\E}{\mathbb{E}}

\DeclareMathOperator*{\Prob}{Pr}

\newcommand{\T}[0]{{\mathrm{T}}}

\newcommand{\sgn}{\mathrm{sgn}}

\newcommand{\pf}{\mathrm{pf}}

\newtheorem{theorem}{Theorem}
\newtheorem{corollary}{Corollary}
\newtheorem{proposition}{Proposition}

\newtheorem{definition}{Definition}
\newtheorem{lemma}{Lemma}
\newtheorem{remark}{Remark}

\begin{document}
\definecolor{navy}{RGB}{46,72,102}
\definecolor{pink}{RGB}{219,48,122}
\definecolor{grey}{RGB}{184,184,184}
\definecolor{yellow}{RGB}{255,192,0}
\definecolor{grey1}{RGB}{217,217,217}
\definecolor{grey2}{RGB}{166,166,166}
\definecolor{grey3}{RGB}{89,89,89}
\definecolor{red}{RGB}{255,0,0}

\preprint{APS/123-QED}

\title{Classical simulation of free-fermionic dynamics and quantum chemistry with magic input}
\author{Changhun Oh}
\email{changhun0218@gmail.com}
\affiliation{Department of Physics, Korea Advanced Institute of Science and Technology, Daejeon 34141, Korea}

\author{Michał Oszmaniec}
\affiliation{Center for Quantum Enabled-Computing, Center for Theoretical Physics of the Polish Academy of Sciences, Al. Lotników 32/46, 02-668 Warsaw, Poland}

\author{Oliver Reardon-Smith}
\affiliation{Center for Quantum Enabled-Computing, Center for Theoretical Physics of the Polish Academy of Sciences, Al. Lotników 32/46, 02-668 Warsaw, Poland}

\author{Zolt\'an Zimbor\'as}
\affiliation{University of Helsinki, Yliopistonkatu 4 00100 Helsinki, Finland}
\affiliation{HUN-REN Wigner Research Centre for Physics, Budapest, Hungary}
\affiliation{Algorithmiq Ltd, Kanavakatu 3C 00160 Helsinki, Finland}

\begin{abstract}
Establishing the precise computational boundary between classically tractable fermionic systems and those capable of genuine quantum advantage is a central challenge in quantum simulation. 
While injecting non-Gaussian ``magic" inputs into free-fermion circuits is widely expected to generate intractable complexity, we identify a physically motivated intermediate regime.
We prove that for block-product paired non-Gaussian fermionic states, essential quantum simulation primitives---transition amplitudes, overlaps, and arbitrary-weight number correlators---can be efficiently approximated to additive error under free-fermionic dynamics. 
This tractability stems from an algebraic reduction that compresses exponentially large multiparticle interference into a single coefficient of a multivariate Pfaffian polynomial.
Because these classical estimators match the intrinsic $O(1/\sqrt{K})$ statistical uncertainty of quantum hardware utilizing $K$ measurement shots, they constitute a practical benchmark.
Building on this foundation, we construct an additive-error estimator for high-weight Wilson observables in the noninteracting quench of recent trapped-ion experiments, providing a rigorous classical benchmark.
Extending this to quantum chemistry, we demonstrate that core overlap-based subroutines for antisymmetrized products of strongly orthogonal geminals admit efficient additive-error Pfaffian-kernel estimators.
Ultimately, these results sharpen the boundary of quantum advantage, establishing that the paired-electron scaffold is dequantized and clarifying where quantum resources are indispensable. 
\end{abstract}

\maketitle

\section{Introduction}
Quantum processors are increasingly deployed as programmable platforms for simulating fermionic many-body dynamics~\cite{arute2020observation, daley2022practical, alam2025fermionic, google2020hartree}. 
This progress sharpens a central question at the heart of fermionic simulation: which physically relevant combinations of evolutions and initial states genuinely evade the best classical algorithms, and which do not. 
We address this boundary in a controlled yet highly non-trivial setting: structured non-Gaussian paired inputs evolved by number-preserving Gaussian dynamics, i.e., the passive sector of fermionic linear optics~(FLO).

FLO occupies a distinctive place at the interface of physics and computational complexity~\cite{terhal2002classical, jozsa2008matchgates, valiant2002quantum, knill2001fermionic, oszmaniec2017universal, hebenstreit2020computational,oszmaniec2022fermion, dias2024classical}. 
For Gaussian inputs (Slater determinants and, more generally, fermionic Gaussian states), transition amplitudes reduce to determinants or Pfaffians and are computable in polynomial time. 
Conversely, augmenting passive FLO with suitably chosen non-Gaussian ``magic'' resources makes the exact calculation of these amplitudes \#P-hard~\cite{ivanov2016computational}. This augmentation is widely believed to yield classically intractable sampling regimes, in close analogy with boson sampling~\cite{aaronson2011computational}, and can even promote it to universal quantum computation once adaptive measurements and feedforward are allowed~\cite{jozsa2008matchgates, brod2011extending}.
This motivates a common expectation that structured non-Gaussian resources may unlock exponential quantum advantage.

\begin{figure*}[t]
\includegraphics[width=460px]{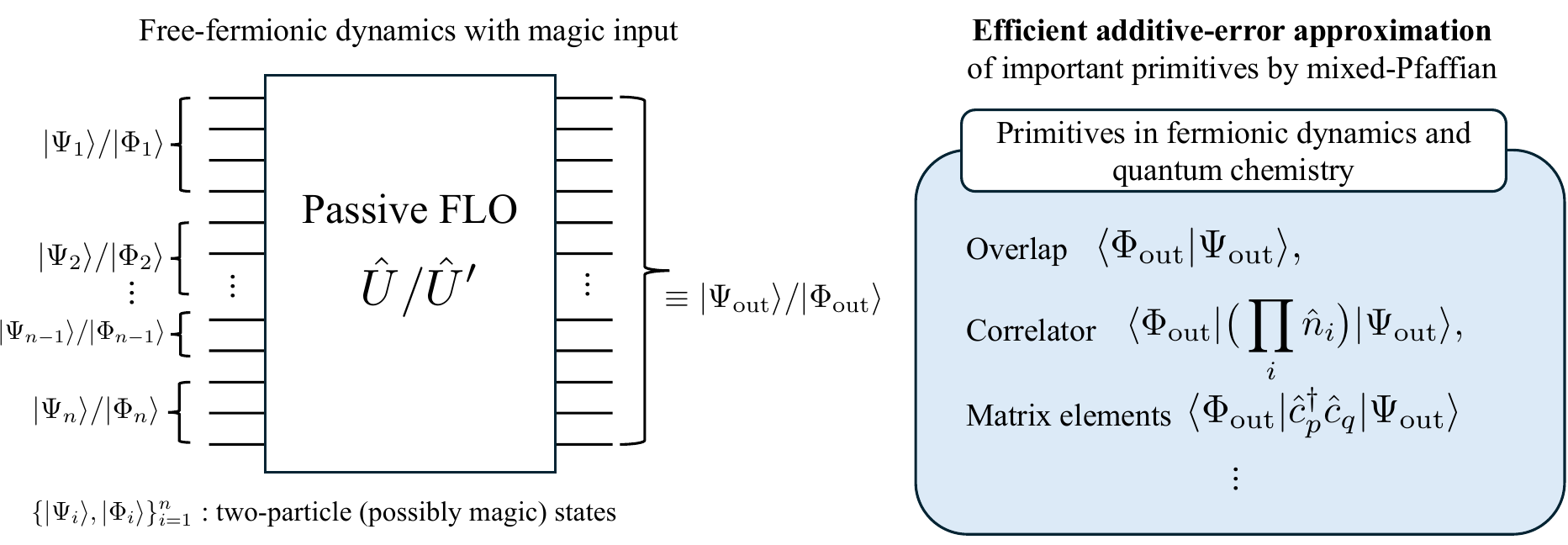}
    \caption{{\bf Efficient classical estimation of non-Gaussian fermionic dynamics.} The schematic illustrates the evolution of a structured non-Gaussian input state—composed of $n$ disjoint two-particle states (paired magic inputs)—under a passive fermionic linear optics~(FLO) unitary $\hat{U}$, which is the main focus of this work. While such dynamics are widely expected to exhibit exponential computational complexity~\cite{oszmaniec2022fermion}, our mixed-Pfaffian reduction structurally bypasses this combinatorial explosion. Equipped with rigorous variance bounds, the framework enables scalable additive-error approximations. Key computational primitives for quantum chemistry and many-body physics, including transition overlaps $\langle\Phi_{\mathrm{out}}|\Psi_{\mathrm{out}}\rangle$, arbitrary-weight number correlators $\langle\Phi_{\mathrm{out}}|(\prod_i \hat{n}_i)|\Psi_{\mathrm{out}}\rangle$, and off-diagonal transition matrix elements $\langle\Phi_{\mathrm{out}}|\hat{c}_p^\dagger \hat{c}_q|\Psi_{\mathrm{out}}\rangle$, are thus rendered classically accessible.}
\label{fig:result}
\end{figure*}

In this work, we challenge the prevalent assumption that the mere presence of such ``magic'' resources necessarily places a fermionic system beyond the reach of classical simulation. 
We reveal a physically motivated middle ground: paired non-Gaussian fermionic states for which the experimentally relevant primitives---transition amplitudes, generic overlaps, and off-diagonal transition matrix elements~\cite{daley2022practical, alam2025fermionic, mcardle2020quantum, huggins2022unbiasing}---can be efficiently approximated to additive error under number-preserving FLO~(see Fig.~\ref{fig:result} for the summary of our main results). 
Crucially, we center our analysis on additive-error estimation as it represents the operationally natural notion of precision for both many-body simulation and quantum chemistry. 
Since physical observables are inherently estimated from a finite number of experimental shots, a classical algorithm attaining equivalent additive precision in polynomial time provides a direct and fair benchmark for quantum hardware.

Our framework relies on an algebraic reduction showing that an exponentially large sum of determinants arising in this setting can be interpreted as a single coefficient of a multivariate Pfaffian polynomial. 
We combine this ``mixed-Pfaffian'' representation with a randomized filtering procedure whose variance is controlled for arbitrary block-product APSG inputs, yielding efficient classical additive-error estimators.

This additive-error tractability yields a series of constructive implications that redefine the boundary of classical benchmarking across different domains. 
(i) In the context of many-body physics, we develop scalable classical estimators for macroscopic diagnostics, such as multi-point number correlators, off-diagonal transition elements, and coarse-grained output statistics that remain tractable even as operator weight grows. 
(ii) We apply these methods to provide rigorous benchmarks for noninteracting quench protocols employed in recent trapped-ion fermionic simulators~\cite{alam2025fermionic}, including a specialized estimator for the experimentally emphasized high-weight Wilson-loop sector.
(iii) For quantum chemistry and electronic-structure workflows~\cite{huggins2022unbiasing, ollitrault2024enhancing}, we demonstrate that prominent quantum subroutines for paired-electron references are effectively dequantized. 
By extending our framework to evaluate transition reduced density matrices~(RDMs) and off-diagonal couplings, we show that the core inner loops of orbital optimization, Hamiltonian matrix evaluation in non-orthogonal configuration interaction~(NOCI)~\cite{baek2023say}, and walker--trial overlaps in auxiliary-field quantum Monte Carlo~(AFQMC) admit efficient additive-error Pfaffian-kernel estimators for antisymmetrized products of strongly orthogonal geminals~(APSG) references. 
By identifying the block-product APSG+passive-FLO framework as classically tractable within the additive-error regime, these results establish that paired-electron structure alone is not sufficient for quantum advantage, implying that any genuine advantage in paired-electron workflows must originate from algorithmic ingredients that go beyond this setting.

\section{Fermionic Fock space and passive linear optics}
\subsection{Fermionic Fock space and exterior algebra}
We briefly present the fermionic Fock space and exterior-algebra notation used throughout this work. For more detailed treatments of fermionic Fock space, exterior algebra, and Gaussian fermionic transformations, see Refs.~\cite{berezin1966method,bravyi2005lagrangian,terhal2002classical}.

Let us consider an $M$-mode fermionic system with annihilation and creation operators $\{\hat{c}_i\}_{i=1}^M$, $\{\hat{c}_i^\dagger\}_{i=1}^M$ obeying the canonical anticommutation relations~\cite{fetter2012quantum}
\begin{align}
    \{\hat{c}_i,\hat{c}_j\}=\{\hat{c}_i^\dagger,\hat{c}_j^\dagger\}=0,\quad
    \{\hat{c}_i,\hat{c}_j^\dagger\}=\delta_{ij}.
\end{align}
The vacuum $|0\rangle$ is defined by $\hat{c}_i|0\rangle=0$ for all $i\in [M]$.
A computational (occupation-number) basis for the Fock space is generated by applying creation operators to the vacuum:
\begin{equation}
    \hat{c}_{i_1}^\dagger \cdots \hat{c}_{i_k}^\dagger |0\rangle, \quad \text{with } 1 \le i_1 < i_2 < \dots < i_k \le M.
\end{equation}
The full fermionic Fock space decomposes into particle-number sectors:
\begin{align}
    \mathcal{F}=\bigoplus_{k=0}^M \mathcal{F}^{(k)},
\end{align}
where $\mathcal{F}^{(k)}$ is the $k$-fermion subspace.

Let \(V \cong \mathbb{C}^M\) be the one-particle space with basis
\(\{e_1,\dots,e_M\}\). We identify the fermionic Fock space with the exterior algebra
$\mathcal{F}(V)=\bigoplus_{k=0}^M \Lambda^k V$ via
\begin{align}
    \hat c^\dagger_{i_1}\cdots \hat c^\dagger_{i_k}\lvert 0\rangle
    \longleftrightarrow
    e_{i_1}\wedge \cdots \wedge e_{i_k},~~~(i_1<\cdots<i_k).
\end{align}
We employ the standard linear map from skew-symmetric matrices to \(2\)-forms,
\begin{equation}
    \alpha(A)\equiv\frac12\sum_{p,q=1}^M A_{pq}  e_p\wedge e_q,
    ~~~ A^\T=-A.
\end{equation}
A \(2\)-form \(\beta\in \Lambda^2 V\) is therefore uniquely represented as
\(\beta=\alpha(W)\) for some skew-symmetric \(W\).
In operator language, this corresponds to the pair-creation operator
\begin{equation}
\beta \longleftrightarrow \hat\beta^\dagger
\equiv \frac12\sum_{p,q=1}^M W_{pq}\hat c^\dagger_p \hat c^\dagger_q.
\end{equation}
Specifically, rank-2 forms serve as the primary building blocks for the paired-input states analyzed in this work. For any single-particle states $u,v\in V$, we define the skew-symmetric matrix $\Omega(u,v)\equiv uv^\T-vu^\T$, which conveniently isolates the rank-2 contribution such that
\begin{equation}
    u\wedge v = \alpha\left(\Omega(u,v)\right).
\end{equation}

\subsection{Passive FLO action and fermionic sampling}
A passive FLO unitary is a number-preserving Gaussian unitary whose Heisenberg action is linear on the mode operators~\cite{terhal2002classical,knill2001fermionic,jozsa2008matchgates}. 
It is specified by a mode unitary \(U\in U(M)\) as
\begin{align}
    \hat{c}_i \mapsto \sum_{j=1}^M U_{ij}\hat{c}_j.
\end{align}
Thus, passive FLO preserves the total fermion number and implements a unitary transformation on the single-particle mode space.

On the one-particle space $V$, $U$ acts linearly; on the $k$-particle sector $\Lambda^k V$, it acts as $\wedge^k U$ (the induced action on exterior powers). In particular, a wedge product transforms as
$(U u_1)\wedge\cdots\wedge(U u_k)  =  (\wedge^k U)(u_1\wedge\cdots\wedge u_k)$.
Likewise, a two-form encoded by $W$ transforms by unitary congruence:
\begin{align}
    W\mapsto U W U^\T.
\end{align}
To maintain notational clarity, we consistently distinguish between second-quantized operators acting on the Fock space, denoted with hats (e.g., $\hat{U}, \hat{\Pi}_T, \hat{D}_T$), and their corresponding induced matrix representations on the single-particle space, denoted without hats (e.g., $U, \Pi_T, D_T$). For instance, the Heisenberg action of a passive FLO unitary is explicitly expressed as $\hat{U}^\dagger \hat{c}_i \hat{U} = \sum_{j=1}^M U_{ij}\hat{c}_j$.

\begin{figure}[t]
\includegraphics[width=200px]{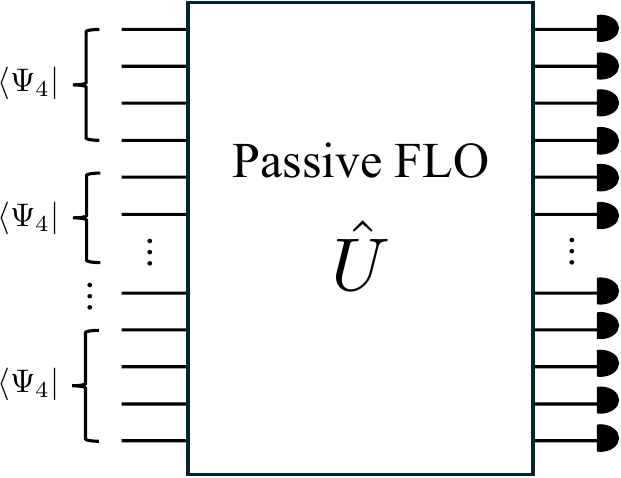}
\caption{Fermionic sampling with $|\Psi\rangle=|\Psi_4\rangle^{\otimes N}$ input with a passive FLO circuit.}
\label{fig:sampling}
\end{figure}

One of the standard computational tasks associated with FLO is fermionic sampling~\cite{oszmaniec2022fermion}~(see Fig.~\ref{fig:sampling}): prepare an input state $|\Psi\rangle$ on $M$ fermionic modes, apply a (passive) FLO unitary $\hat U$, and measure in the occupation-number basis. 
A canonical non-Gaussian ``magic'' input for this task is the four-mode
two-fermion block
\begin{align}
    |\Psi_4\rangle \equiv \frac{1}{\sqrt 2}(|1100\rangle+|0011\rangle),
\end{align}
and we use the product state \(|\Psi_4\rangle^{\otimes N}\) as our running example below. 
The output is an occupation pattern $x\in\{0,1\}^M$ distributed according to
\begin{align}
    p_{\Psi,U}(x)=|\langle x|\hat U|\Psi\rangle|^2.
\end{align}
Equivalently, the fundamental primitive is the transition amplitude $\langle \Phi|\hat U|\Psi\rangle$ since sampling probabilities and many expectation values reduce to sums of such amplitudes (or their products) (see Sec.~\ref{sec:estimation}).

While this primitive is efficiently computable for Gaussian inputs~\cite{terhal2002classical, valiant2002quantum, jozsa2008matchgates,
knuth1995overlapping}, generic non-Gaussian ``magic'' inputs can render its exact evaluation \#P-hard~\cite{ivanov2016computational, oszmaniec2022fermion}.
Our focus in the following sections is to identify a structured paired subfamily, containing the running example $|\Psi_4\rangle^{\otimes N}$, for which the same amplitude nonetheless admits an efficient {\it additive-error} classical {\it approximation}. 
Concretely, the inputs we consider are block-product paired states, where each block contributes exactly one fermion pair drawn from a superposition of mutually disjoint pair slots~(see Def.~\ref{def:disjointBGP}). 
This paired structure allows us to rewrite the multiparticle amplitude as a mixed-Pfaffian coefficient and estimate it efficiently via randomized sign-averaging.

\section{Additive-error estimators for transition amplitudes}\label{sec:main}
\subsection{Mixed Pfaffian coefficient extraction}
Because transition amplitudes between fermionic basis states represent antisymmetrized multiparticle overlaps, they are naturally governed by matrix determinants~\cite{ivanov2016computational}. 
A useful geometric identity is that if $r_1,\dots,r_N\in \mathbb{C}^N$ are row vectors, their determinant isolates the top-form coefficient of their wedge product:
\begin{align}\label{eq:det}
    \det
    \begin{pmatrix} 
        r_1 \\ \vdots \\ r_N 
    \end{pmatrix}=
    [r_1\wedge r_2\wedge\cdots\wedge r_N]_{\mathrm{vol}}.
\end{align}
Here, $\mathrm{vol}=e_1\wedge\cdots\wedge e_N$ denotes the canonical volume form, and the linear functional $[\cdot]_{\mathrm{vol}}$ extracts the scalar coefficient of this volume form.

When this geometric intuition is extended to pair-creation operators and entangled fermionic pairs, Pfaffians emerge as the foundational algebraic structure~\cite{ivanov2016computational}.
For a $2N\times 2N$ anti-symmetric matrix $A$, the Pfaffian is defined as~\cite{knuth1995overlapping}
\begin{align}
    \pf(A)\equiv\frac{1}{2^N N!}\sum_{\sigma\in \mathcal{S}_{2N}}\sgn(\sigma)\prod_{i=1}^N A_{\sigma(2i-1),\sigma(2i)},
\end{align}
where $\mathcal{S}_{2N}$ is the symmetric group on $2N$ elements.
Crucially, the Pfaffian is mathematically equivalent to the top-form coefficient of the $N$-th exterior power of its associated 2-form~\cite{knuth1995overlapping}:
\begin{align}\label{eq:pf_wedge}
    \pf(A)
    =\frac{1}{N!}\left[\alpha(A)\wedge \cdots \wedge \alpha(A)\right]_\text{vol}=\frac{1}{N!}\left[\alpha(A)^{\wedge N}\right]_\text{vol}.
\end{align}
This identity (derived explicitly in App.~\ref{app:pf_identity_proof}) explains why Pfaffians are the native algebraic objects when dealing with products of pair-creation operators (2-form).
Furthermore, much like determinants, the Pfaffian of a generic skew-symmetric matrix can be computed efficiently in $O(N^3)$ time~\cite{wimmer2012algorithm}.

A central technical observation of this work is that the exponentially large sum of determinant terms, arising when a blockwise pair-superposition state evolves under passive FLO, can be expressed as a single coefficient of a multivariate Pfaffian polynomial. 
This observation forms the backbone of our additive-error estimators. 
Having expressed the sum in this way, we extract the polynomial coefficient using a technique similar to those previously developed for estimating the matrix permanent and mixed discriminants~\cite{gurvits2005complexity, aaronson2011computational, aaronson2012generalizing, oh2024quantum}, adapting them to the fermionic domain.

\begin{figure}[t]
\includegraphics[width=\columnwidth]{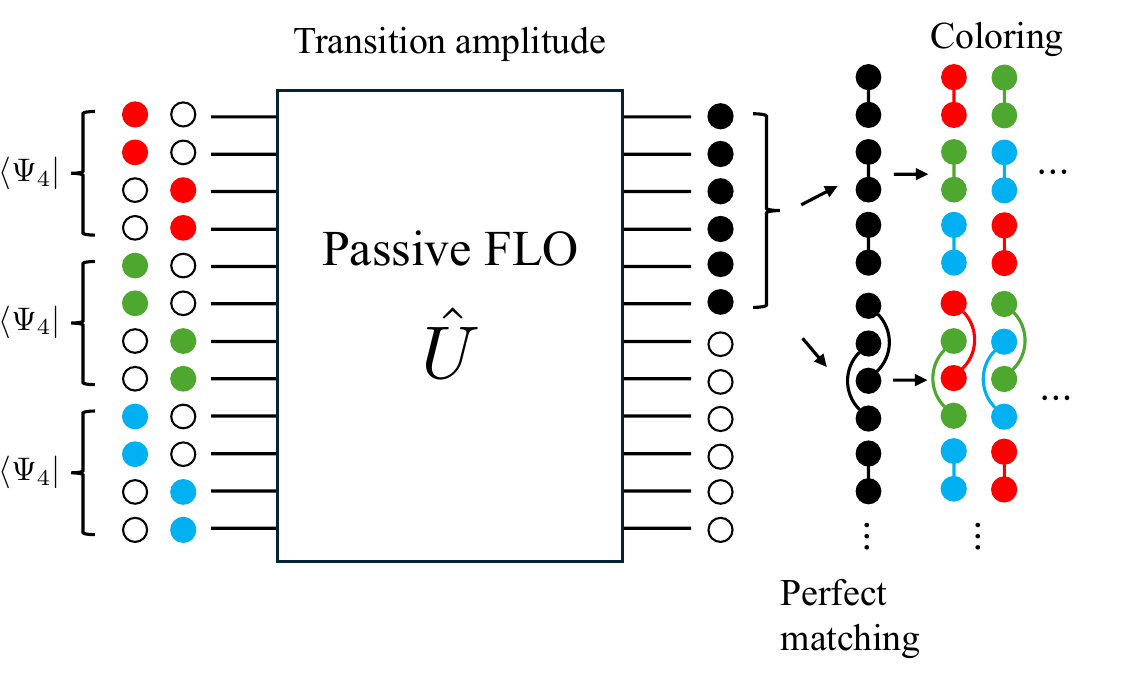}
\caption{
{\bf Mixed-Pfaffian coefficient as colored perfect matchings.}
A block-product paired input is evolved by a passive FLO unitary \(\hat U\) and projected onto a fixed \(2N\)-fermion output pattern. 
After restricting to the occupied output modes, each input block contributes a skew-symmetric matrix \(B_t\). 
The Pfaffian \(\pf(\sum_t y_t B_t)\) expands as a signed sum over perfect matchings of the \(2N\) output vertices, while the variables \(y_t\) color the matched edges by their originating block. 
Extracting the multilinear coefficient \([y_1\cdots y_N]\) keeps exactly the colorings that use each block once, thereby enforcing the paired-input constraint.
}
\label{fig:mixed_pfaffian_intuition}
\end{figure}

To formalize this, let $B_1,\dots,B_N \in \mathbb{C}^{2N\times 2N}$ be skew-symmetric matrices representing the transformed two-forms corresponding to $N$ disjoint blocks. We construct the multivariate Pfaffian polynomial:
\begin{equation}
    P(y_1,\dots,y_N) \equiv \mathrm{pf}\left(\sum_{t=1}^N y_t B_t\right).
\end{equation}
Rather than evaluating the full polynomial, we target its multilinear coefficient, sometimes referred to as the \emph{mixed Pfaffian} in natural analogy to the mixed discriminant~\cite{ikai2011theory}: 
\begin{equation}
    \left[\prod_{t=1}^N y_t\right] P(y_1,\dots,y_N),
\end{equation}
where $[\mathcal{M}]P$ denotes the coefficient of the monomial $\mathcal{M}$ in the polynomial $P$.

The coefficient-extraction identity also admits a simple graph-theoretic interpretation, which is useful for understanding why both the Pfaffian and the block labels \(y_t\) appear naturally; see Fig.~\ref{fig:mixed_pfaffian_intuition}.
For a fixed \(2N\)-particle output configuration, a Pfaffian expands as a signed sum over perfect matchings of the \(2N\) occupied output modes. 
Replacing the single skew matrix by the formal sum \(\sum_t y_t B_t\) promotes each edge of a matching to a colored edge, where the color \(t\) records which input block contributed that pair. 
Thus, \(P(y)=\pf(\sum_t y_t B_t)\) is a generating function for edge-colored perfect matchings. 
The one-pair-per-block constraint is then imposed by extracting the multilinear coefficient \([y_1\cdots y_N]\): this keeps precisely those colored matchings in which every color appears once and discards all terms in which some block is used multiple times or not at all. 
The random sign average used below is therefore not an additional approximation, but an exact Fourier filter implementing this colored-matching projection.

The same filtering statement can be written directly in exterior-algebra language. 
By exploiting the multilinearity of the 2-form mapping \(\alpha(\cdot)\) and the Pfaffian-wedge identity, Eq.~\eqref{eq:pf_wedge}, the colored-matching projection above is equivalently expressed as
\begin{equation}\label{eq:mixed1}
    \left[\prod_{t=1}^N y_t\right] \mathrm{pf} \left(\sum_{t=1}^N y_t B_t\right) = \left[\alpha(B_1) \wedge \cdots \wedge \alpha(B_N)\right]_{\mathrm{vol}}.
\end{equation}
Thus, the multilinear coefficient rigorously isolates the specific multiparticle interference pathway in which exactly one fermion pair is selected from each block.
This relation translates the problem of summing exponentially many physical determinantal configurations into evaluating a single coefficient of a bounded algebraic generating function.

\subsection{Illustrative example: paired 4-mode blocks}\label{sec:psi4}

To illustrate the core mechanism of our framework, we first analyze the canonical input $|\Psi_4\rangle$ introduced above~\cite{ivanov2016computational, oszmaniec2022fermion}.
For a global system of $4N$ modes, we consider the tensor product state $|\Psi\rangle = |\Psi_4\rangle^{\otimes N}$ evolving under a passive FLO unitary $\hat U$. Although a direct expansion of the transition amplitude $\langle x|\hat U|\Psi\rangle$ yields an exponentially large sum of determinants over the $2^N$ block configurations, the underlying paired structure allows us to mathematically compress this sum into a single mixed-Pfaffian coefficient. By isolating this coefficient via a discrete randomized sign-averaging procedure, we reduce the exact multiparticle amplitude calculation to the Monte Carlo estimation of a bounded Pfaffian-valued random variable.

Note that this specific paired configuration is precisely the canonical ``magic'' state family utilized in recent theoretical proofs to establish the worst-case hardness of non-Gaussian fermionic sampling~\cite{oszmaniec2022fermion}, not an arbitrary toy example. Furthermore, it is mathematically identical to the exact tensor-product spin-triplet dimer state deployed as the initial non-Gaussian resource in recent large-scale trapped-ion quantum simulators~\cite{alam2025fermionic} (see Sec.~\ref{sec:phasecraft}). Therefore, identifying an efficient additive-error estimator for this exact state carries direct physical and complexity-theoretic implications.

To construct this estimator, let $U\in \mathbb{C}^{4N\times 4N}$ be the single-particle unitary matrix corresponding to $\hat U$. Expanding the block superposition for the transition amplitude $A \equiv \langle x|\hat U|\Psi\rangle$ yields a sum over $2^N$ block-selection patterns. Let $i = (i_1, \dots, i_N) \in \{0,1\}^N$ encode whether block $t$ contributes the configuration $|1100\rangle$ ($i_t=0$) or $|0011\rangle$ ($i_t=1$), and let $R(i) \subset [4N]$ denote the resulting set of occupied input modes (exactly two per block). The transition amplitude expands as a sum of Slater determinants:
\begin{equation}
    A = 2^{-N/2} \sum_{i \in \{0,1\}^N} \det(U_{S_x, R(i)}),
\end{equation}
where $S_x \subset [4N]$ is the set of occupied output modes in $|x\rangle$.

To circumvent the exponential cost of evaluating $2^N$ determinants, we reorganize this structured sum. 
Define the output-restricted transpose isometry $X \equiv (U_{S_x,:})^\T = (U^\T)_{:,S_x}\in \mathbb{C}^{4N \times 2N},$
which satisfies \(X^\dagger X=I_{2N}\).
Let $v_k^\T \equiv e_k^\T X$ so that $v_k \in \mathbb{C}^{2N}$ is a column vector.
For each block $t \in \{1, \dots, N\}$, we construct the $2N \times 2N$ skew-symmetric matrix:
\begin{equation}
    B_t \equiv \Omega(v_{4t-3},v_{4t-2})+\Omega(v_{4t-1},v_{4t}).
\end{equation}
Consider the multivariate Pfaffian polynomial $P(y_1, \dots, y_N) \equiv \mathrm{pf}\left(\sum_{t=1}^N y_t B_t\right)$. Combining Eqs.~\eqref{eq:det} and \eqref{eq:mixed1}, we arrive at a key algebraic identity of our framework
\begin{align}
    S \equiv [y_1 \cdots y_N] P(y_1,\dots,y_N)=\sum_{i \in \{0,1\}^N} \det(U_{S_x, R(i)}).
\end{align}
Hence, the amplitude collapses to $A = 2^{-N/2} S$.

To estimate $S$ efficiently, we sample random variables $b = (b_1, \dots, b_N) \sim \mathrm{Unif}(\{\pm 1\}^N)$ and define the single-shot random variable:
\begin{equation}
    \mathcal{Z}(b) \equiv \mathrm{pf}\left(\sum_{t=1}^N b_t B_t\right) \prod_{t=1}^N b_t.
\end{equation}
By discrete Fourier isolation, the exact coefficient is recovered by the expectation value $S = \mathbb{E}_b[\mathcal{Z}(b)]$. This formulation immediately suggests a Monte Carlo estimator: draw $K$ independent samples $b^{(1)}, \dots, b^{(K)}$, compute the respective $O(N^3)$ Pfaffians, and calculate the sample average $\tilde{S} \equiv \frac{1}{K} \sum_{k=1}^K \mathcal{Z}(b^{(k)})$.

Hence, it now suffices to analyze the bound of the random variable $\mathcal{Z}(b)$ to find the error bound of the estimator. 
As proven using operator-norm inequalities in App.~\ref{app:general_estimator} (Lemma~\ref{lem:pf_bound_apsg}), 
$\left|\pf\!\left(\sum_{t=1}^N b_t B_t\right)\right|\le 1$ for all sign configurations $b \in \{\pm 1\}^N$, guaranteeing $|\mathcal{Z}(b)| \le 1$.
Since $|\mathcal{Z}(b)|\le 1$, applying Hoeffding's inequality implies that it suffices to take $K=O(\epsilon^{-2}\log\delta^{-1})$ to achieve additive error $\epsilon$ with failure probability at most $\delta$.
Equivalently, defining $\tilde{A} \equiv 2^{-N/2}\tilde{S}$, we obtain $|\tilde{A}-A|\le \epsilon/\sqrt{2^N}$ with probability at least $1-\delta$.

\begin{theorem}[Additive-error transition amplitude for block-magic input]\label{thm:fock-block-magic-overlap}
    Let $\hat{U}$ be a passive FLO unitary on $4N$ modes, let $|\Psi\rangle = |\Psi_4\rangle^{\otimes N}$ with $|\Psi_4\rangle = (|1100\rangle + |0011\rangle)/\sqrt{2}$, and let $|x\rangle$ be any Fock basis state with $2N$ fermions. There exists a classical randomized algorithm that outputs an estimator $\tilde{A}$ satisfying $\mathrm{Pr}[|\tilde{A} - \langle x|\hat{U}|\Psi\rangle| > \epsilon / \sqrt{2^N}] \le \delta$. The algorithm requires a sample complexity of $K=O(\epsilon^{-2}\log\delta^{-1})$, with each sample processed in $O(N^3)$ time.
\end{theorem}
We emphasize that while this additive error, $\epsilon/\sqrt{2^N}$, is exponentially small in the system size, it does not imply a polynomial-time algorithm for computing the multiplicative approximations necessary to resolve worst-case \#P-hard exact sampling probabilities~\cite{oszmaniec2022fermion}.

Using the above estimator, we can also efficiently approximate the transition
amplitude from $|\Psi_4\rangle^{\otimes N}$ to itself through a passive FLO
$\hat U$. Let $s\in\{0,1\}^N$ label the bra-side block configuration, with
$s_t=0$ selecting the local pair $|1100\rangle$ and $s_t=1$ selecting
$|0011\rangle$, and let $|R(s)\rangle$ be the corresponding Fock basis state.
Then
\begin{align}
    \langle\Psi_4^{\otimes N}|\hat{U}|\Psi_4^{\otimes N}\rangle
    =
    2^{-N}\sum_{s\in\{0,1\}^N}
    \mathbb{E}_b[\mathcal{Z}(b;s)]
    =
    \mathbb{E}_{s,b}[\mathcal{Z}(b;s)],
\end{align}
where $s$ is sampled uniformly over the $2^N$ block configurations.
Since the same bound applies to the random variable $\mathcal{Z}$, we have the theorem:
\begin{theorem}[Additive-error transition amplitude for block-magic overlap]\label{thm:block-magic-overlap}
    Let $\hat{U}$ be a passive FLO unitary on $4N$ modes, and let $|\Psi\rangle = |\Psi_4\rangle^{\otimes N}$ with $|\Psi_4\rangle = (|1100\rangle + |0011\rangle)/\sqrt{2}$.
    Consider the transition amplitude $A=\langle\Psi|\hat{U}|\Psi\rangle$. There exists a classical randomized algorithm that outputs an estimator $\tilde{A}$ satisfying $\Pr[|\tilde{A}-A|>\epsilon]\le\delta$. The algorithm requires a sample complexity $K=O(\epsilon^{-2}\log\delta^{-1})$, with each sample processed in $O(N^3)$ time.
\end{theorem}

\subsection{Extension to block-product APSG states}\label{sec:general}
While the $|\Psi_4\rangle^{\otimes N}$ block-magic state illustrates the core mixed-Pfaffian mechanism, we now generalize this framework to a broader and physically ubiquitous family of non-Gaussian paired states.
In this section, we focus on the block-product APSG structure, i.e., tensor products of \emph{single-pair} geminals on disjoint blocks. This keeps the physically relevant feature that each block contributes exactly one electron pair, while allowing arbitrary block sizes and arbitrary superpositions over disjoint pair slots within each block.

\begin{definition}[Block-product APSG states]\label{def:disjointBGP}
Partition the $M$ fermionic modes into $N$ disjoint blocks.
On each block $t$, let $r_t$ denote the number of mutually disjoint \emph{pair slots}.
Fix disjoint pair creators $\{\hat p_{t,j}^\dagger\}_{j=1}^{r_t}$, where
$\hat p_{t,j}^\dagger \equiv \hat c_{t,2j-1}^\dagger \hat c_{t,2j}^\dagger$.
Define the localized geminal operator
\begin{align}
    \hat{\eta}_t^\dagger \equiv \sum_{j=1}^{r_t} w_{t,j}\,\hat p_{t,j}^\dagger,
    \qquad w_{t,j}\in\mathbb{C},
\end{align}
where the block normalization $\sum_{j=1}^{r_t}|w_{t,j}|^2=1$ is assumed. 
A block-product APSG state is
\begin{align}
    |\Psi\rangle \equiv \bigotimes_{t=1}^N |\psi_t\rangle
    = \prod_{t=1}^N \hat{\eta}_t^\dagger |0\rangle,
\end{align}
which has exactly $2N$ fermions in total (one pair per block).
\end{definition}

Since local phases can be absorbed into surrounding passive FLO maps, we henceforth assume $w_{t,j}\ge 0$ without loss of generality. We will also employ the notation $[r_t] = \{1,2,\hdots r_t\}$. We prove direct analogues of Theorems~\ref{thm:fock-block-magic-overlap} and~\ref{thm:block-magic-overlap} for block-APSG states. Note that a quantum device estimating the same bounded quantity by repeated measurements requires $k=O(\epsilon^{-2})$ shots for additive error $\epsilon$; our classical estimators from Theorems~\ref{thm:overlap_fock_apsg} and~\ref{thm:overlap_apsg} match this operational scaling.


\begin{theorem}[Additive-error Fock-APSG overlap estimator]
\label{thm:overlap_fock_apsg}
    Let $\hat U$ be a passive FLO unitary on $M$ modes, let $|\Psi\rangle$ be an arbitrary block-product APSG state as in Def.~\ref{def:disjointBGP}, and let $\ket{x}$ be a Fock state. Define $\gamma = \prod_{t=1}^N \max_{j\in [r_t]} w_{t,j} $ and note that $\gamma \leq 1$ for normalized $\ket{\Psi}$. Then there exists a classical randomized algorithm that outputs an estimator $\tilde A$ for $\langle x|\hat U|\Psi\rangle$ satisfying $\Prob\!\left[|\tilde A-\langle x|\hat U|\Psi\rangle|>\epsilon \gamma\right]\le \delta$ with sample complexity $K=O(\epsilon^{-2}\log\delta^{-1})$.
\end{theorem}

Following a methodology analogous to the $|\Psi_4\rangle$ case, the exponential sum over physical pair configurations can be mathematically compressed into a single multilinear coefficient of a multivariate Pfaffian polynomial. 
Thus, by applying the same discrete sign-averaging technique, we obtain an unbiased Monte Carlo estimator for the fixed-output transition amplitude (we defer the full mixed-Pfaffian coefficient identity to App.~\ref{app:general_estimator}). 
For Theorem~\ref{thm:overlap_fock_apsg}, the remaining step is a pointwise bound on this fixed-output estimator; after rescaling by $\gamma=\prod_t\max_j w_{t,j}$, Lemma~\ref{lem:pf_bound_apsg} gives the required uniform control.

We then pass from fixed Fock outputs to arbitrary APSG bra states in Theorem~\ref{thm:overlap_apsg}. 
Here, the factor $\gamma$, which plays the same role as the factor $1/\sqrt{2^N}$ in Theorem~\ref{thm:fock-block-magic-overlap}, is generally too large for the simple averaging argument used in Theorem~\ref{thm:block-magic-overlap} to apply directly. 
Instead, App.~\ref{app:apsg_overlap_proofs} proves a bounded second moment for the combined bra-sampling and mixed-Pfaffian estimator, yielding an efficient estimator for general block-product APSG transition amplitudes.

\begin{theorem}[Additive-error APSG overlap estimator]
\label{thm:overlap_apsg}
    Let $\hat U$ be a passive FLO unitary on $M$ modes, and let $|\Phi\rangle$ and $|\Psi\rangle$ be arbitrary block-product APSG states as in Def.~\ref{def:disjointBGP}. Then there exists a classical randomized algorithm that outputs an estimator $\tilde A$ for $\langle \Phi|\hat U|\Psi\rangle$ satisfying $\Prob\!\left[|\tilde A-\langle \Phi|\hat U|\Psi\rangle|>\epsilon\right]\le \delta$ with sample complexity $K=O(\epsilon^{-2}\log\delta^{-1})$.
\end{theorem}
While we defer the full proof of Theorem~\ref{thm:overlap_apsg} to App.~\ref{app:apsg_overlap_proofs}, we note here that the estimator is conceptually simple given the result of Theorem~\ref{thm:overlap_fock_apsg}. We rewrite
\begin{align}
    \ket{\Phi} = \bigotimes_{t=1}^N\left(\sum_{j=1}^{s_t} v_{t,j}\,\hat{q}_{t,j}^\dagger\ket{0}\right)
    = \bigotimes_{t=1}^N\left(\sum_{j=1}^{s_t} \abs{v_{t,j}}^2\, \frac{\hat{q}_{t,j}^\dagger\ket{0}}{v_{t,j}^*}\right),
\end{align}
where $s$, $v$ and $\hat{q}$ are the analogues of $r$, $w$ and $\hat{p}$ for the state $\ket{\Phi}$.
Now we sample unnormalized Fock states $(v_{t,j}^*)^{-1}\hat p_{t,j}^\dagger\ket{0}$ from the product probability distribution $x\mapsto \prod_{t=1}^N \abs{v_{t,x_t}}^2$, where $x$ takes values in $[s_1]\times [s_2]\times \hdots \times [s_N]$. We combine the sample average over sampled $x$ vectors with the result of Theorem~\ref{thm:overlap_fock_apsg} to obtain an unbiased estimator for the full transition amplitude. 
In App.~\ref{app:apsg_overlap_proofs}, we show that the variance of this combined estimator is bounded by $1$, which leads to efficient classical estimation, for example via the median-of-means procedure.

\subsection{Extension to non-unitary Gaussian maps}\label{sec:nonunitary}
The polynomial-time efficiency and variance bounds established above apply to unitary passive FLOs ($\|U\|_\mathrm{op}=1$). 
However, our mixed-Pfaffian approach relies solely on the multilinearity of the exterior algebra, specifically through the congruence transformation of two-forms ($W \mapsto G W G^\T$) and the discrete Fourier filtering.
Consequently, the coefficient-extraction identity underlying our estimator remains exact for any number-preserving Gaussian map $\hat G$
with single-particle matrix $G\in\mathbb{C}^{M\times M}$. For instance, for every fixed-output Fock state $|x\rangle$ with $2N$ fermions,
\begin{equation}
    \langle x|\hat G|\Psi\rangle=\mathbb{E}_{b\sim\mathrm{Unif}(\{\pm1\}^N)}\left[\mathcal{Z}(b;x)\right],
\end{equation}
with $\mathcal{Z}(b;x)$ defined in App.~\ref{app:general_estimator}.

While the algebraic extraction remains exact per sample, the loss of unitarity degrades the analytical bounds governing the Monte Carlo sampling variance.

\begin{theorem}[Additive-error APSG overlap estimator for non-unitary Gaussian maps]
\label{thm:overlap_apsg_nonunitary}
    Let $\hat G$ be a non-unitary number-preserving fermionic Gaussian map on $M$ modes with single-particle matrix $G$, and let $|\Phi\rangle$ and $|\Psi\rangle$ be arbitrary block-product APSG states as in Definition~\ref{def:disjointBGP}. Then there exists a classical randomized algorithm outputting an estimator $\tilde A$ for $\langle \Phi|\hat G|\Psi\rangle$ such that $\Prob\!\left[|\tilde A-\langle \Phi|\hat G|\Psi\rangle|>\epsilon\right]\le \delta$ with sample complexity $K=O\!\left(\|G\|_{\mathrm{op}}^{4N}\epsilon^{-2}\log\delta^{-1}\right)$.
\end{theorem}

The proof is identical in spirit to the unitary case~(see App.~\ref{app:apsg_overlap_proofs}): we combine the fixed-output mixed-Pfaffian estimator with an $L^2$ importance-sampling reduction, and use the non-unitary fixed-output bound $|\mathcal Z(b;x)|\le \|G\|_{\mathrm{op}}^{2N}$, which yields $\mathbb E[|\mathcal X|^2]\le \|G\|_{\mathrm{op}}^{4N}$.

We emphasize that this norm-dependent overhead is not an artifact of the Pfaffian reduction itself. Rather, it exposes the variance barrier induced by non-unitary Gaussian propagation, the same obstruction that appears in real-time Hubbard--Stratonovich~(HS) dynamics and is closely related to the sign/phase problem in AFQMC imaginary-time projection.

\section{Estimation of macroscopic observables}\label{sec:estimation}

\subsection{Multipoint number correlators}\label{sec:num_corr}
The overlap estimators developed above provide the basic transition-kernel primitive. 
We now show that the same primitive also yields additive-error estimators for physically relevant observables, including multipoint number correlators and transition matrix elements between distinct many-body states. 
These quantities are central to extracting many-body diagnostics, such as density profiles, doublon correlations, and spin/charge structure factors, and to evaluating Hamiltonian matrix elements in algorithms such as subspace expansions~\cite{mcclean2017hybrid}, linear response theory~\cite{stanton1993equation}, and NOCI~\cite{sundstrom2014non, huggins2020non}.

Let $S \subseteq [M]$ be a target subset of modes, and define the multipoint number correlator 
\begin{equation}
    \hat{O}_S \equiv \prod_{i \in S} \hat{n}_i, \quad \text{where } \hat{n}_i \equiv \hat{c}_i^\dagger \hat{c}_i.
\end{equation}
For Gaussian dynamics, low-weight instances of these number correlators can be evaluated exactly by propagating Majorana operators in the Heisenberg picture~\cite{miller2025simulation, d2025majorana}. 
The cost of this exact approach, however, scales as $O(M^{|S|})$, which makes high-weight number correlators and large-scale number-correlation functions increasingly costly.
Our result addresses precisely the finite-precision task relevant for quantum simulation experiments: additive-error estimation.
In this operational regime, the operator-weight bottleneck is avoided by rewriting products of number operators as randomized averages of bounded Gaussian parity kernels.

To establish a general algorithmic primitive, we consider the transition matrix element of $\hat{O}_S$ between two independently generated states: $\langle\Phi| \hat{U}_L^\dagger \hat{O}_S \hat{U}_R |\Psi\rangle$, where $|\Phi\rangle$ and $|\Psi\rangle$ are block-product APSG states, and $\hat{U}_L$ and $\hat{U}_R$ are passive FLO unitaries. Evaluating this multipoint transition amplitude naively requires an exponentially large Hilbert space of Fock outcomes. 
However, because the eigenvalues of $\hat{n}_i$ on Fock basis states are Boolean ($n_i \in \{0,1\}$), we invoke the exact operator identity $\hat{n}_i = (1 - (-1)^{\hat{n}_i})/2$ to systematically expand the multipoint correlator into a sum of parity strings:
\begin{equation} \label{eq:correlator_expansion}
    \hat{O}_S = 2^{-|S|} \prod_{i \in S} \left(1 - (-1)^{\hat{n}_i}\right) = 2^{-|S|} \sum_{T \subseteq S} (-1)^{|T|} \hat{\Pi}_T,
\end{equation}
where $\hat{\Pi}_T \equiv \prod_{i \in T} (-1)^{\hat{n}_i}$ is the parity operator on subset $T$. 
Taking the transition matrix element yields:
\begin{equation}
    \langle\hat{O}_S\rangle_{L,R} \equiv \langle\Phi| \hat{U}_L^\dagger \hat{O}_S \hat{U}_R |\Psi\rangle = \mathbb{E}_{T \subseteq S}[(-1)^{|T|} \mu_T^{L,R}],
\end{equation}
where $\mu_T^{L,R} \equiv \langle\Phi|\hat{U}_L^\dagger \hat{\Pi}_T \hat{U}_R|\Psi\rangle.$
Importantly, a parity operator is a diagonal passive FLO: it flips the sign of the modes in $T$ and leaves all other modes unchanged. 
Hence, $\hat U_L^\dagger \hat\Pi_T\hat U_R$ is again a passive FLO, so $\mu_T^{L,R}$ is exactly the overlap primitive treated above. 
Therefore, an unbiased estimator is obtained by sampling $T\subseteq S$ uniformly and then drawing one mixed-Pfaffian sign sample for the corresponding passive-FLO kernel. 
The combined one-shot variable inherits the same second-moment bound as the
underlying overlap estimator, so a median-of-means estimator gives sample
complexity independent of $|S|$.

\begin{theorem}[Additive-error estimation of generic transition number correlators]\label{thm:cor}
    Let $\hat{U}_L$ and $\hat{U}_R$ be passive FLOs on $M$ modes, and let $|\Phi\rangle$ and $|\Psi\rangle$ be arbitrary block-product APSG states. For any targeted subsystem $S \subseteq [M]$, there exists a classical randomized algorithm that outputs an estimator $\widetilde{\langle\hat{O}_S\rangle}_{L,R}$ for the transition correlator $\langle\hat{O}_S\rangle_{L,R} \equiv \langle\Phi|\hat{U}_L^\dagger \hat{O}_S \hat{U}_R|\Psi\rangle$ satisfying $\mathrm{Pr}[| \widetilde{\langle\hat{O}_S\rangle}_{L,R} - \langle\hat{O}_S\rangle_{L,R} | > \epsilon] \le \delta$. The algorithm requires a sample complexity of $K=O(\epsilon^{-2}\log\delta^{-1})$.
\end{theorem}

\subsection{Off-diagonal transition elements and reduced density matrices}\label{sec:rdm}
While the parity-string expansion resolves diagonal multi-point correlators, evaluating off-diagonal, non-commuting observables—such as transition one- and two-body reduced density matrices (RDMs)—is equally fundamental, particularly for quantum chemistry and linear response theory.

Remarkably, off-diagonal transition elements can be generated from the same
Gaussian transition kernel by introducing a formal complex one-body source.
For number-preserving Gaussian maps \(\hat G_L\) and \(\hat G_R\), define
\begin{equation}
    \mathcal K_{\Phi,\Psi}(J;\hat G_L,\hat G_R)
    \equiv
    \langle\Phi|
    \hat G_L
    \exp(\hat{\bm c}^\dagger J\hat{\bm c})
    \hat G_R
    |\Psi\rangle,
    ~
    J\in\mathbb C^{M\times M}.
\end{equation}
Here \(J\) is a formal source matrix; the deformed operator is a
number-preserving Gaussian map, unitary only when \(J\) is anti-Hermitian. Since
our mixed-Pfaffian construction applies to number-preserving Gaussian maps, the
kernel is evaluated by replacing the single-particle matrix by
\(G_L e^J G_R\).

Transition RDMs are obtained by differentiating this analytic kernel at
\(J=0\). For example,
\begin{equation}
    \langle\Phi|
    \hat G_L \hat c_p^\dagger \hat c_q \hat G_R
    |\Psi\rangle
    =
    \left.
    \frac{\partial}{\partial J_{pq}}
    \mathcal K_{\Phi,\Psi}(J;\hat G_L,\hat G_R)
    \right|_{J=0}.
\end{equation}
Higher constant-order transition RDM elements are obtained from ordered
products of such one-body sources; see App.~\ref{app:chem_sources}.

\begin{corollary}[Additive-error estimation of transition RDMs]\label{cor:rdm}
    For any passive FLO unitary $\hat U$ and arbitrary block-product APSG states $|\Phi\rangle$ and $|\Psi\rangle$, any constant-order transition RDM element can be estimated to additive error $\epsilon$ in polynomial time by evaluating the local derivatives of the mixed-Pfaffian generating function.
\end{corollary}

\subsection{Marginal probabilities and binned distributions for fermionic outputs}
The same transition-kernel primitive also gives access to coarse-grained output statistics. 
This is distinct from estimating individual output probabilities, which may be exponentially small and are not the natural target of additive-error simulation. 
Instead, many experimentally relevant measurements ask for marginal probabilities or binned event weights, obtained by summing over exponentially many microscopic Fock outcomes. 
We show that such quantities reduce to diagonal passive-FLO overlap kernels and therefore fall within the mixed-Pfaffian additive-error framework.

Fix a subset $S \subseteq [M]$ and a target bit pattern $a \in \{0,1\}^{|S|}$. 
The marginal probability $P_S(a) \equiv \sum_{x: x_S = a} p_{\Psi,U}(x)$ expands as a sum of parity strings:
\begin{equation}
    P_S(a) = \mathbb{E}_{T \sim \mathrm{Unif}(\mathcal{P}(S))} \left[ (-1)^{\sum_{i \in T} a_i} \mu_T \right],
\end{equation}
where $\mu_T \equiv \langle \Psi | \hat U^\dagger \hat \Pi_T \hat U | \Psi \rangle$.
Thus, $P_S(a)$ admits additive-error estimation by the same Monte Carlo strategy and concentration bounds derived in Theorem~\ref{thm:cor}, with sample complexity independent of the number of unobserved modes.

More broadly, it is often useful to coarse-grain the exponentially large outcome space by grouping outputs according to a macroscopic linear statistic. This viewpoint is particularly transparent through \emph{characteristic functions}: in molecular vibronic-spectra simulation, discrete Fourier techniques map binned distributions to diagonal-phase overlap kernels that can be estimated efficiently~\cite{oh2024quantum}.
Motivated by the same principle, we fix an integer weight vector $\omega \in \mathbb{Z}_{\ge 0}^M$ and define the binned probability distribution $G(\Omega) \equiv \sum_{x: \omega \cdot x = \Omega} p_{\Psi,U}(x)$, where $\Omega \in \{0, 1, \dots, \Omega_{\max}\}$ and $\Omega_{\max} \equiv \sum_{i=1}^M \omega_i$.
To circumvent the sum over an exponentially large pre-image, we map the distribution to frequency space via its discrete Fourier transform (DFT), $\tilde{G}(k) \equiv \sum_{\Omega=0}^{\Omega_{\max}} G(\Omega) e^{i k \theta \Omega}$, where $\theta \equiv 2\pi / (\Omega_{\max} + 1)$.
This characteristic function can be analytically recast as the expectation of a diagonal unitary:
\begin{align}
    \tilde{G}(k) 
    &= \sum_{x \in \{0,1\}^M} p_{\Psi, U}(x)  e^{i k \theta \omega \cdot x} \\
    &= \langle\Psi| \hat{U}^\dagger e^{i k \theta \sum_{i=1}^M \omega_i \hat{n}_i} \hat{U} |\Psi\rangle.
\end{align}
The operator in the middle simply applies the phase \(e^{i k\theta\omega_i}\) to mode \(i\), and is therefore a diagonal passive-FLO unitary with single-particle matrix $D(k)=\mathrm{diag}(e^{i k\theta\omega_1},\dots,e^{i k\theta\omega_M}).$
Consequently, the Fourier coefficient $\tilde{G}(k) = \langle\Psi|\hat{U}^\dagger \hat{D}(k) \hat{U}|\Psi\rangle$ is identically a passive-FLO overlap computable within the same mixed-Pfaffian primitive discussed above.
The exact macroscopic distribution $G(\Omega)$ is reconstructed via the inverse DFT: $G(\Omega) = \frac{1}{\Omega_{\max}+1} \sum_{k=0}^{\Omega_{\max}} e^{-i k \theta \Omega} \tilde{G}(k)$.
When $\Omega_{\max}$ is polynomially bounded, the full binned distribution can be reconstructed with polynomial overhead. Estimating each Fourier coefficient to additive precision $\epsilon$ gives each reconstructed bin additive error at most $\epsilon$ by the inverse-DFT triangle bound, with only the standard logarithmic overhead for simultaneous control over all bins~\cite{oh2024quantum}.

This characteristic-function viewpoint parallels earlier DFT-based approaches to coarse-grained output statistics, including Gaussian/bosonic sampling settings and molecular vibronic spectra~\cite{seron2024efficient,lim2025efficient,oh2024quantum}. 
Our result gives the fermionic counterpart: for paired inputs with definite pair numbers under passive FLO, macroscopic binned statistics reduce to diagonal passive-FLO overlaps and hence to the mixed-Pfaffian additive-error primitive.

Beyond physical diagnostics, these efficient coarse-graining techniques have direct implications for generative quantum machine learning, particularly for fermionic Born machines~\cite{bako2025fermionic}. 
Training these generative models typically requires estimating distance metrics such as the maximum mean discrepancy~(MMD)~\cite{rudolph2024trainability}, which relies heavily on evaluating multi-point marginals and overlap kernels. 
By rendering these expectation values classically tractable, our framework provides a scalable method for evaluating the MMD loss, enabling the robust classical pre-training of high-dimensional paired-fermionic models prior to quantum deployment~\cite{recio2025train}.

\section{Classical simulation of trapped-ion fermionic dynamics}\label{sec:phasecraft}
The operational power of the mixed-Pfaffian framework becomes most transparent when applied to recent experimental milestones that probe fermionic dynamics beyond the reach of exact state-vector simulation. 
A prominent example is the large-scale trapped-ion experiment by Alam et al.~\cite{alam2025fermionic}, in which a highly structured, non-Gaussian initial state is evolved and interrogated using both local diagnostics and macroscopic Wilson-loop observables. 
Because such experiments inject extensive non-Gaussian ``magic'' into otherwise controllable fermionic dynamics, they are naturally viewed as candidates for regimes beyond classical simulation.

Here, we show that a substantial part of this benchmarking regime is nevertheless classically accessible in additive error. 
In particular, the noninteracting benchmark, including the high-weight doublon Wilson-loop sector, can be estimated in additive error by our mixed-Pfaffian estimators, while interactions introduce a distinct non-unitary variance barrier.

Consider the spinful single-band Fermi--Hubbard model studied in Ref.~\cite{alam2025fermionic},
\begin{equation}
    \hat{H}=\hat{H}_0+\hat{H}_W,
\end{equation}
where
\begin{align}
    \hat{H}_0\equiv -J\sum_{\langle i,j\rangle,\sigma}(e^{i\phi_{ij}}\hat c_{i\sigma}^\dagger \hat c_{j\sigma}+\mathrm{h.c.}),~
    \hat{H}_W\equiv W\sum_i \hat n_{i\uparrow}\hat n_{i\downarrow}.
\end{align}
Here, $\langle i,j\rangle$ ranges over all nearest-neighbor links and $\sigma\in \{\uparrow, \downarrow\}$ represents spins.
The real-time state is $\ket{\Psi(t)}=e^{-i\hat H t}\ket{\Psi_0}$. We identify a fermionic mode with a spin-orbital $(i,\sigma)$, so that $M=2L$ for $L$ lattice sites.

The experiment reports several macroscopic diagnostics, including local charge densities, connected spin correlators, and doublon/triplet populations. Concretely, we use
\begin{align}
    n_{i\sigma}(t) &\equiv \langle \hat n_{i\sigma}(t)\rangle,\\
    C_{zz}(i,j;t) &\equiv 4\Big(\langle \hat S_i^z(t)\hat S_j^z(t)\rangle-\langle \hat S_i^z(t)\rangle\langle \hat S_j^z(t)\rangle\Big),\\
    N_d(t) &\equiv \sum_{i}\langle \hat n_{i\uparrow}(t)\hat n_{i\downarrow}(t)\rangle,\\
    n_{\mathrm{triplets}}(t) &\equiv \frac{2}{L}\sum_{\langle i,j\rangle} C_{zz}(i,j;t),
\end{align}
where $\hat S_i^z\equiv \tfrac12(\hat n_{i\uparrow}-\hat n_{i\downarrow})$ (see App.~\ref{app:phasecraft_model} for the experimental normalization conventions and additional diagnostics).

\begin{figure*}[t]
\includegraphics[width=480px]{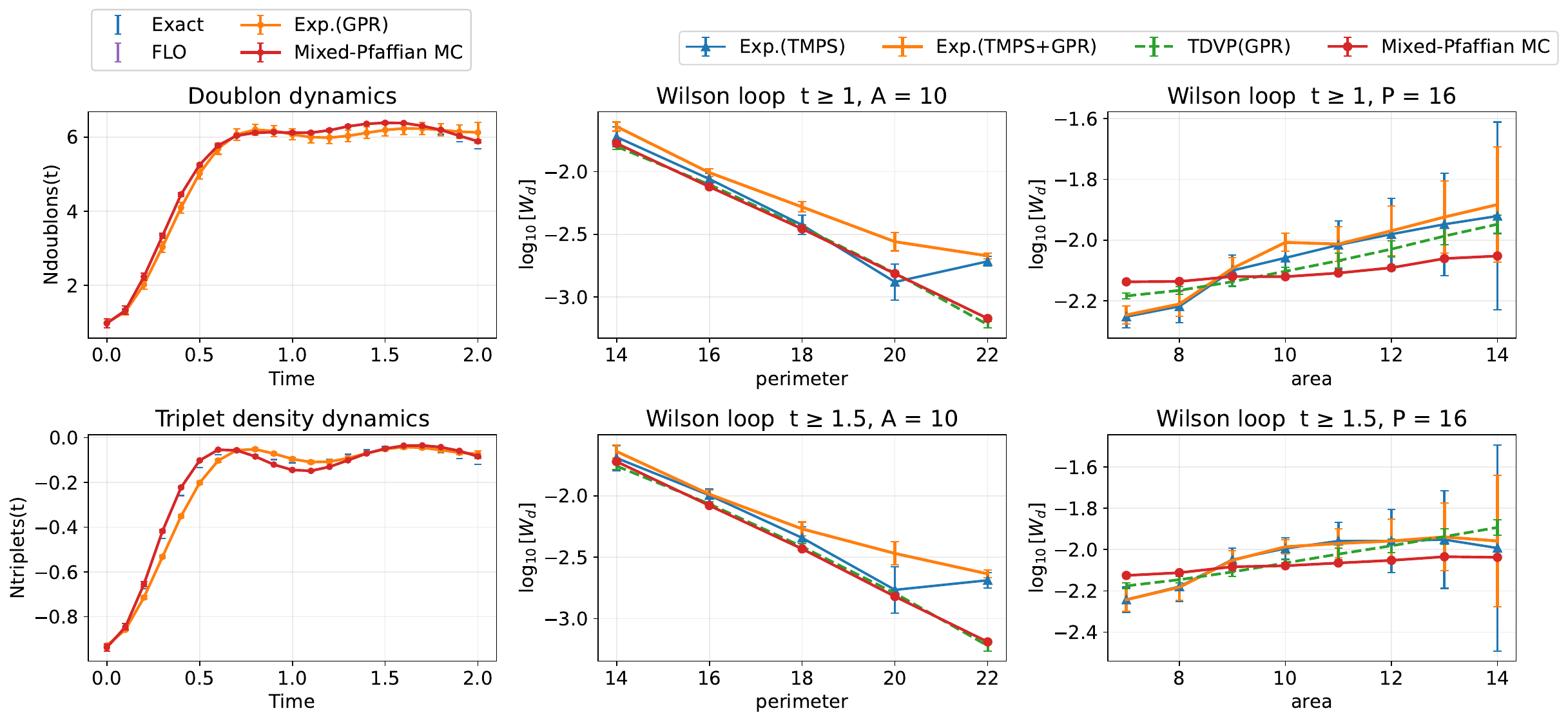}
\caption{\textbf{Noninteracting benchmark: doublon dynamics, triplet density, and Wilson loops.}
Left column: low-weight observables, namely the doublon number $N_d(t)$ and the nearest-neighbor triplet diagnostic $n_{\mathrm{triplets}}(t)$, for which our Monte Carlo results are benchmarked against exact, FLO, and experimentally reported references.
Right two columns: high-weight doublon Wilson-loop observables in the two standard cuts used in Ref.~\cite{alam2025fermionic}, shown as functions of perimeter at fixed area and as functions of area at fixed perimeter, with the two time windows indicated in each panel.
The Wilson-loop data are obtained from the specialized charge-phase estimator described in Sec.~\ref{sec:phasecraft} and App.~\ref{app:phasecraft_model}.
Error bars on our Monte Carlo points are obtained from seed-to-seed fluctuations across 10 independent runs, with $10^3$ samples per run.
Method labels for the experimental and tensor-network references follow Ref.~\cite{alam2025fermionic}.
For the low-weight panels, ``Exact'' denotes the exact untrotterized free-fermion evolution, ``FLO'' denotes the corresponding Trotterized FLO reference, and ``Exp.(GPR)'' denotes the experimentally measured data after the mitigation and Gaussian-process-regression smoothing procedure of Ref.~\cite{alam2025fermionic}.
For the Wilson-loop panels, ``Exp.(TMPS)'' denotes the experimentally inferred device signal corrected by training with short-time MPS/TDVP data, ``Exp.(TMPS+GPR)'' denotes Gaussian-process-regression smoothing applied to the TMPS-corrected experimental data, and ``TDVP+GPR'' denotes the tensor-network reference obtained from TDVP/MPS simulations and smoothed for comparison. 
``Mixed-Pfaffian MC'' denotes our additive-error estimator.}
\label{fig:noninteracting}
\end{figure*}

The experiment initializes the system in a deterministic product of nearest-neighbor spin-triplet dimers~(equivalent to $|\Psi_4\rangle$ state),
\begin{align}
    |T_0\rangle_{j,k}=\frac{1}{\sqrt{2}}\left(\hat c_{j\uparrow}^\dagger \hat c_{k\downarrow}^\dagger+\hat c_{j\downarrow}^\dagger \hat c_{k\uparrow}^\dagger\right)\ket{0},
\end{align}
possibly with additional localized holons and doublons. 
Consequently, the overall initial state falls directly into the scope of our polynomial-time overlap estimators (Theorem~\ref{thm:overlap_apsg}).

This observation makes the noninteracting benchmark, $W=0$, especially diagnostic.
In this regime, since the time evolution is passive FLO, the only apparent source of hardness is the extensive paired magic in the initial state and the complexity of the measured observable.
Low-weight quantities such as $N_d(t)$ and the nearest-neighbor correlators entering $n_{\mathrm{triplets}}(t)$ can already be checked by standard Gaussian/Majorana propagation~\cite{miller2025simulation,d2025majorana}.
They therefore serve as a calibration layer for our implementation.

The more stringent test is the high-weight diagonal string sector emphasized in Ref.~\cite{alam2025fermionic}.
These Wilson-line/loop observables are products of local holon or doublon projectors along extended contours, and are not naturally handled by low-weight Majorana propagation.
Here, we focus on the doublon Wilson loop $\hat W_C=\prod_{i\in C}(1-\hat n_{i\uparrow}\hat n_{i\downarrow})$.
Although this observable is diagonal in the occupation basis, it is not simply a multipoint number correlator of the form \(\prod_i \hat n_i\).
We therefore use a specialized two-phase charge decomposition of the local projector \(1-\hat n_{i\uparrow}\hat n_{i\downarrow}\), rather than the general number-correlator estimator of Sec.~\ref{sec:num_corr}.
This decomposition converts each Wilson-loop sample into a diagonal passive-FLO overlap, which is then evaluated by the same mixed-Pfaffian primitive as in Theorem~\ref{thm:overlap_apsg}.

More precisely, the resulting one-shot random variable is bounded by $(2/\sqrt{3})^{|C|}$, which implies the contour-dependent additive-error sample complexity
\begin{equation}
    K = O\!\left(
    \left(\frac{4}{3}\right)^{|C|}
    \epsilon^{-2}\log\delta^{-1}
    \right).
\end{equation}
This is the estimator used for the Wilson-loop panels in Fig.~\ref{fig:noninteracting}. The detailed construction is given in App.~\ref{app:phasecraft_model}.


Figure~\ref{fig:noninteracting} summarizes the two levels of our noninteracting benchmark. The left column shows low-weight observables, for which the mixed-Pfaffian Monte Carlo agrees with the available exact, FLO, and experimentally reported references. The right two columns address the high-weight doublon Wilson-loop sector, which was previously treated as a regime where no reliable exact classical reference was available. In this sector, our estimator provides an independent additive-error classical benchmark. The resulting Wilson-loop values are statistically well converged: the seed-to-seed Monte Carlo fluctuations are visibly smaller than the separation between the experimental error-mitigated curves and the tensor-network-based predictions.

Within the noninteracting benchmark, the mixed-Pfaffian results are broadly consistent with the tensor-network trend, especially in the fixed-area perimeter scans, rather than with the larger values inferred from the error-mitigated experimental/GPR curves. We therefore do not interpret the Wilson-loop discrepancy in the $W=0$ data as evidence of a failure of classical tensor-network methods. Instead, our results show that this high-weight Wilson sector remains classically benchmarkable in the noninteracting regime, and that the observed differences are not explained by Monte Carlo sampling noise in our estimator. This strengthens the diagnostic role of the $W=0$ case: it can now be used not only for low-weight observables, but also as a stringent benchmark for the high-weight Wilson-loop post-processing pipeline.

In our numerical benchmark, the largest Wilson contours considered in Fig.~\ref{fig:noninteracting} have perimeter $|C|=22$, and are already well converged with $K=10^3$ samples. Taking this as an empirical calibration point, the expected sample count for a larger contour of perimeter $|C|$ scales as $K(|C|) \approx 10^3\left(\frac{4}{3}\right)^{|C|-22}$.
Since $(4/3)^8\simeq 10$, every additional eight sites along the contour costs roughly one extra order of magnitude in samples. Thus, at comparable precision, contours with perimeters $|C|\simeq 54$, $62$, and $70$ would require approximately $10^7$, $10^8$, and $10^9$ samples, respectively. 
This indicates that the experimentally relevant Wilson contours studied here remain comfortably accessible to the specialized classical estimator, while substantially larger contours would eventually become practically demanding due to the exponential $(4/3)^{|C|}$ scaling.

While the noninteracting regime admits this efficient additive-error estimation, extending this simulation to interacting dynamics ($W>0$) via Trotterization and real-time HS transformations inherently introduces non-unitary propagators. As mathematically formalized in Sec.~\ref{sec:nonunitary}, this loss of unitarity triggers an exponential variance barrier in the Monte Carlo sampling, obstructing scalable classical simulation at long times. We provide a comprehensive analysis of this interacting regime, including the precise sample-complexity bounds and the rigorous worst-case envelope for the real-time variance barrier, in App.~\ref{app:interacting_dynamics}.

\section{Implications for quantum chemistry}
\label{sec:chem}

\subsection{Paired-electron references as a chemically motivated scaffold}

Entangled electron-pair (geminal) wave functions provide a foundational, physically intuitive route to capturing strong many-body entanglement—often referred to as \emph{static} correlation in chemistry. This is particularly crucial in near-degeneracy regimes such as molecular bond breaking, where single-determinant approximations fundamentally fail~\cite{tecmer2022geminalreview}. Within this family, the APSG class serves as a paradigmatic, size-extensive scaffold. Deeply connected to broader models like generalized valence bond (GVB) and seniority-zero architectures, it represents a universally recognized structural backbone for simulating strongly correlated systems~\cite{surjan2012sog, tecmer2022geminalreview}. Recently, these paired-electron ans\"atze have also re-emerged as highly efficient, hardware-aware building blocks for near-term quantum processors (e.g., oo-upCCD circuits)~\cite{zhao2023oo_upccd,zhao2024pt2_oo_upccd}.

Our results establish a crucial structural baseline for quantum advantage claims hinging on this ubiquitous scaffold. 
Because the algebraic reduction to mixed-Pfaffian coefficients is exact for all block-product APSG states, the exponentially large multiparticle-interference expansion can be reorganized into a compact coefficient-estimation problem. 
Rather than resolving the worst-case \#P-hardness of exact evaluation, our framework bypasses this barrier at the operational level by targeting the additive-error overlap and transition primitives relevant to finite-shot quantum workflows. 
Combined with the sampling bounds of our framework, these primitives can be classically approximated to additive error throughout the block-product APSG scaffold.
Consequently, quantum hardware should not be tasked merely with representing paired-electron correlations; rather, quantum resources should be focused on resolving specific ``beyond-pair'' correlations that break this classically benchmarkable scaffold.

\subsection{Redefining chemistry workflows as Gaussian transition primitives}
\label{sec:chem:primitives}
Electronic-structure algorithms---whether classical or quantum---typically avoid explicit simulation of the exponentially large Fock space by repeatedly invoking a small set of algebraic primitives.
Because the standard molecular Hamiltonian comprises at most two-body interactions (quartic in fermionic operators), the transition energy $H_{\Phi,\Psi} \equiv \langle \Phi|\hat H|\Psi\rangle$ between distinct many-body states is completely determined by the unnormalized transition one- and two-body RDMs~\cite{helgaker2013molecular, mazziotti2012two}:
\begin{align}
    \gamma_{pq}^{\Phi,\Psi} &\equiv \langle \Phi|\hat c_p^\dagger \hat c_q|\Psi\rangle,~~~\Gamma_{pq,rs}^{\Phi,\Psi} &\equiv \langle \Phi|\hat c_p^\dagger \hat c_q^\dagger \hat c_s \hat c_r|\Psi\rangle.
    \label{eq:chem:rdm_energy}
\end{align}
Furthermore, the quartic Coulomb interaction can often be reorganized using
low-rank or Cholesky factorizations into a sum of squares of one-body
operators~\cite{motta2019efficient}. In this form,
\begin{equation}
    \hat H
    =
    E_0+\hat H_1+\frac12\sum_{\ell}\lambda_\ell \hat O_\ell^2,
    \label{eq:chem:sum_of_squares}
\end{equation}
where $E_0$ is a scalar energy shift, $\hat H_1=\sum_{pq}h_{pq}\hat c_p^\dagger\hat c_q$ is a one-body operator, and each $\hat O_\ell=\sum_{pq}(L_\ell)_{pq}\hat c_p^\dagger\hat c_q$ is a one-body operator arising from the low-rank factorization of the two-electron integrals.
Thus, evaluating molecular Hamiltonian matrix elements reduces to computing transition overlaps and transition moments generated by one-body operators, such as $\langle \Phi|\hat O|\Psi\rangle$ and
$\langle \Phi|\hat O_\ell^2|\Psi\rangle$, or more generally their Gaussian-source counterparts $\langle \Phi|\hat U_L^\dagger \hat O \hat U_R|\Psi\rangle$.

While evaluating these off-diagonal Gaussian kernels is classically trivial for single-determinant inputs, injecting correlated references is widely presumed to invoke exponential computational hardness. However, by restricting the reference to the strongly orthogonal APSG scaffold, this computational bottleneck is bypassed at the level of the Gaussian transition-kernel primitive. Because our simulation primitive evaluates kernels under an arbitrary passive FLO map $\hat U$, any single-particle basis rotation used to parameterize the APSG state can be mathematically absorbed into $\hat U$. By invoking Corollary~\ref{cor:rdm}, the expectation values of these squared operators---and hence the full quartic Hamiltonian matrix element---directly reduce to derivatives of the Gaussian transition kernel, computable via our framework.

\begin{remark}[Chemistry primitives as Gaussian transition kernels]\label{rem:chem_apsg_kernel}
    In the chemistry applications below, the required quantum subroutines are not arbitrary many-body amplitudes. Orbital rotations, NOCI couplings, AFQMC walkers, and one-body source insertions all reduce to number-preserving Gaussian transition kernels $\langle\Phi|\hat G|\Psi\rangle$ or to derivatives of such kernels.
    Thus, the mixed-Pfaffian primitive applies directly to the overlap, transition-RDM, and local-energy building blocks. The unitary cases inherit the bounded-variance additive-error guarantee, whereas non-unitary walkers retain the exact reduction but acquire the operator-norm variance barrier of Sec.~\ref{sec:nonunitary}.
\end{remark}

\subsection{Classical boundaries of paired-reference workflows}
\label{sec:chem:benchmarkability}

\textbf{Expanding the classically tractable optimization set for QPE initial states.}
The efficiency of quantum phase estimation~(QPE) is fundamentally bottlenecked by the initial overlap $p_0=|\langle \Psi_{\mathrm{init}}|\Psi_0\rangle|^2$ between the prepared state $|\Psi_\mathrm{init}\rangle$ and the target ground state $|\Psi_0\rangle$. To alleviate the exponential decay of $p_0$ with system size, recent proposals advocate variationally optimizing the single-particle basis to minimize the reference energy prior to QPE~\cite{ollitrault2024enhancing}. However, to maintain classical tractability, such optimizations have been restricted to single Slater determinant references. As explicitly noted in Ref.~\cite{ollitrault2024enhancing}, expanding this variational basis optimization to multi-determinant ensembles---which is crucial for capturing strong correlation and increasing $p_0$---is generally presumed to be classically intractable, as the evaluation cost grows exponentially with the number of determinants.

Our mixed-Pfaffian framework analytically circumvents this computational barrier, substantially expanding the classically tractable optimization manifold. By upgrading the initial reference to a paired APSG state, we incorporate a compact but strongly correlated multi-determinant ensemble while preserving the rigid one-pair-per-block scaffold. The variational energy objective becomes
\begin{equation}
    E_{\mathrm{APSG}}(\kappa)=\langle \Psi_{\mathrm{APSG}}| \hat U^\dagger(\kappa) \hat H \hat U(\kappa) |\Psi_{\mathrm{APSG}}\rangle,
\label{eq:chem:qpe_energy_apsg}
\end{equation}
where $\hat U(\kappa) = \exp(\hat\kappa)$ is an orbital-rotation unitary generated by the anti-Hermitian one-body operator $\hat\kappa = \sum_{p>q} \kappa_{pq} (\hat c_p^\dagger \hat c_q - \hat c_q^\dagger \hat c_p)$ with real parameters $\kappa_{pq}$. 
To minimize this objective and obtain the optimal basis, the algorithm requires evaluating the energy and its orbital gradients. The local gradient with respect to the matrix element $\kappa_{pq}$ at $\kappa=0$ is exactly given by:
\begin{align}
    \left. \frac{\partial E_{\mathrm{APSG}}(\kappa)}{\partial \kappa_{pq}} \right|_{\kappa=0}
    &=
    \langle \Psi_{\mathrm{APSG}} |
    [\hat H,\hat c_p^\dagger \hat c_q - \hat c_q^\dagger \hat c_p]
    | \Psi_{\mathrm{APSG}} \rangle.
    \label{eq:chem:qpe_grad}
\end{align}
Because this commutator is again a polynomial in creation and annihilation operators of degree at most four, Eq.~\eqref{eq:chem:qpe_grad} reduces to a linear combination of transition 1-RDM and 2-RDM elements.
By Corollary~\ref{cor:rdm}, these RDM elements can be efficiently approximated to additive error for arbitrary block-product APSG references.
Accordingly, the energy and orbital-gradient primitives entering APSG pre-optimization are classically accessible within our framework throughout the block-product APSG manifold.

\medskip

\textbf{Dequantizing the overlap primitive in imaginary-time projection~(AFQMC).}
Auxiliary-field quantum Monte Carlo~(AFQMC) evaluates ground-state properties by projecting the system in imaginary time. This is achieved by decoupling the two-body interactions via HS transformations, which generate an ensemble of stochastic, non-interacting imaginary-time ``walker'' states~(see Ref.~\cite{huggins2022unbiasing} for the details),
\begin{equation}
|\Phi(\sigma)\rangle \equiv \hat B(\sigma_N)\cdots \hat B(\sigma_1) |\Phi_0\rangle,
\label{eq:chem:afqmc_walker}
\end{equation}
where each $\hat B(\sigma_t)$ is a generally non-unitary number-preserving Gaussian operator determined by the auxiliary field configuration $\sigma_t$. 
Because this HS transformation inevitably introduces a severe sign (or phase) problem, practical AFQMC imposes constraints guided by a trial state $|\Psi_T\rangle$. This requires evaluating the overlap weight $W(\sigma)\equiv \langle \Psi_T|\Phi(\sigma)\rangle$ and the local energy $E_L(\sigma)$:
\begin{equation}
E_{\mathrm{mixed}}
\approx
\frac{\mathbb{E}_{\sigma} \left[ W(\sigma) E_L(\sigma) \right]}{\mathbb{E}_{\sigma} \left[ W(\sigma) \right]},
\label{eq:chem:afqmc_mixed}
\end{equation}
where $E_L(\sigma)\equiv \langle \Psi_T|\hat H|\Phi(\sigma)\rangle / \langle \Psi_T|\Phi(\sigma)\rangle$.

For simple Slater-determinant trial states, these overlap and local-energy primitives are classically inexpensive. 
The motivation for quantum-assisted AFQMC arises when one seeks to use a more strongly correlated, non-Gaussian trial state $|\Psi_T\rangle$, for which the corresponding overlap and transition-RDM evaluations are widely presumed to be classically difficult~\cite{huggins2022unbiasing}. 
A quantum processor can, in principle, be used to estimate this inner-loop primitive. 
However, the outer stochastic average over auxiliary-field trajectories $\sigma$ remains a classical AFQMC sampling problem, and the sign/phase-problem overhead associated with the non-unitary HS propagation is not removed merely by evaluating the trial overlap on a quantum device. 
Thus, in this setting, the operational role of the quantum subroutine is specifically to supply the trial-state overlap and local-energy kernels.

Our mixed-Pfaffian framework provides a classical additive-error route to this inner-loop primitive for APSG trial states.
Taking $|\Psi_T\rangle=|\Psi_{\mathrm{APSG}}\rangle$ provides a correlated multi-determinant paired reference, but the walker $|\Phi(\sigma)\rangle$ is still generated by a number-preserving Gaussian map $\hat G(\sigma)$. 
Therefore, the overlap $W(\sigma)=\langle \Psi_{\mathrm{APSG}}|\hat G(\sigma)|\Phi_0\rangle$ reduces to the non-unitary APSG transition primitive of Theorem~\ref{thm:overlap_apsg_nonunitary}, and the local energy $E_L(\sigma)$ is obtained from the corresponding transition RDMs via Corollary~\ref{cor:rdm}. 
The non-unitary propagation contributes the same operator-norm variance envelope identified in Sec.~\ref{sec:nonunitary}, but the overlap/RDM evaluation itself is classically accessible by mixed-Pfaffian coefficient extraction. 
Consequently, within the APSG scaffold, the AFQMC overlap and local-energy inner loops can be benchmarked classically using our method. Quantum advantage in this setting would therefore need to rely on trial states or dynamics that go beyond the paired APSG geometry.

\medskip

\textbf{Subspace expansions and non-orthogonal configuration interaction (NOCI).}
Beyond diagonal properties, evaluating off-diagonal transition elements between macroscopically distinct, non-orthogonal many-body states is a central challenge in algorithms such as subspace expansions and NOCI~\cite{mcclean2017hybrid}. These algorithms require evaluating off-diagonal overlaps and Hamiltonian couplings:
\begin{equation}
    S_{LR}=\langle \Psi_L|\Psi_R\rangle,\qquad
    H_{LR}=\langle \Psi_L|\hat H|\Psi_R\rangle.
\label{eq:chem:offdiag_SH}
\end{equation}
Because these matrix elements couple two independently parameterized manifolds, they cannot in general be recast as a simple expectation value. 
This apparent structural complexity is the core justification for quantum eigensolvers targeting non-orthogonal quantum eigensolvers~(NOQE)~\cite{baek2023say}, where the claimed quantum advantage explicitly relies on the presumed classical hardness of evaluating $S_{LR}$ and $H_{LR}$.

Remarkably, directly leveraging the transition correlators formulated in Sec.~\ref{sec:num_corr}, this intuition breaks down for paired-electron reference states. 
Suppose the non-orthogonal basis states are independently optimized paired references, $|\Psi_L\rangle = \hat U_L |\Psi_0\rangle$ and $|\Psi_R\rangle = \hat U_R |\Psi_0\rangle$, generated by applying distinct basis-rotation passive FLOs to the same underlying definite-pair base state. 
The off-diagonal overlap collapses into a single transition element:
\begin{equation}
    S_{LR} = \langle \Psi_0| \hat U_L^\dagger \hat U_R |\Psi_0\rangle.
\end{equation}
Because the product $\hat U_{\mathrm{eff}} \equiv \hat U_L^\dagger \hat U_R$ is simply another number-preserving Gaussian unitary, this ostensibly problematic off-diagonal overlap falls exactly into the purview of Theorem~\ref{thm:overlap_apsg}, admitting an efficient additive-error randomized evaluation within the paired-input tractable regime developed above.
Simultaneously, via Corollary~\ref{cor:rdm}, the Hamiltonian couplings $H_{LR} = \langle \Psi_0| \hat U_L^\dagger \hat H \hat U_R |\Psi_0\rangle$ are extracted with no additional algorithmic overhead.

\medskip
\textbf{Redrawing the boundary of quantum hardness.}
This broad classical tractability necessitates a sharper delineation of the true boundary of quantum advantage. While matchgate circuits augmented with fermionic magic inputs are computationally hard in the worst case~\cite{leimkuhler2025exponential}, we constructively carve out a broad, chemically relevant tractable region within this landscape. 
Because our exact mixed-Pfaffian extraction systematically bypasses the combinatorial explosion of the paired-electron scaffold, any genuine exponential quantum advantage must originate from algorithmic ingredients that rigorously break this structural geometry. 
Ultimately, this dictates that future quantum resources should focus on highly entangled regimes definitively beyond the APSG and passive-FLO manifold---such as unstructured multi-particle entanglement, active non-Gaussian unitary layers, or particle-number-violating Bogoliubov transformations.

\section{Conclusion}
In this work, we have identified a substantial and physically ubiquitous intermediate regime in the complexity landscape of FLO. We demonstrated that for a broad family of structured non-Gaussian states—specifically those constructed as tensor products of pair-creation operators—fundamental algorithmic primitives remain classically tractable to additive error under passive FLO evolution.

By geometrically mapping the exterior-algebra expansion of multiparticle amplitudes into a mixed-Pfaffian coefficient extraction problem, we developed a bounded Monte Carlo estimator driven by discrete root-of-unity Fourier filtering. This approach not only refines the traditional ``Gaussian-easy versus magic-hard" complexity dichotomy, but it translates a theoretical representation into a highly practical algorithmic toolkit. 
Our framework enables scalable verification of near-term fermionic quantum simulators using macroscopic observables and establishes a rigorous additive-error classical benchmark for the block-product APSG paired scaffold.

Our algebraic reduction also reveals a sharp structural boundary for this exact geometrical compression. The efficiency of our framework relies on the mathematical geometry of $2$-forms (pair-creation operators), which compress into computable matrix Pfaffians. Extending this initial resource to generalized $k$-particle clustering ($k \ge 3$, such as forms representing quartet correlations) fundamentally breaks this specific geometry. Evaluating the exact multiparticle interference for such higher-order forms maps mathematically to computing hyper-Pfaffians (or equivalently, counting perfect matchings in $k$-uniform hypergraphs), which is \#P-hard (VNP-complete)~\cite{valiant1979complexity, ikenmeyer2019hyperpfaffians}. Thus, the exact Pfaffian-based compression developed here is specific to the paired-electron scaffold, even though other approximation strategies may still be possible in more restricted higher-cluster regimes.

This structural boundary naturally highlights several compelling directions for future research. A primary open challenge is extending these mixed-Pfaffian extraction techniques to active FLO (Bogoliubov transformations) that violate particle-number conservation. Because Bogoliubov dynamics entangle different particle-number sectors, such an extension will require novel algebraic adjustments to the discrete phase-averaging filtration that currently underpins our lossless compression. 
Furthermore, identifying whether analogous algebraic geometries or generalized randomized estimators can be constructed for broader classes of weakly coupled clusters---or whether such structures mark a regime where genuinely quantum resources become necessary---remains a central open question in defining the ultimate reach of classical many-body simulation.

\begin{acknowledgements}
We would like to thank the Phasecraft team for discussion about their experiment.
CO was supported by the National Research Foundation of Korea Grants (No. RS-2024-00431768 and No. RS-2025-00515456) funded by the Korean government (Ministry of Science and ICT (MSIT)) and the Institute of Information \& Communications Technology Planning \& Evaluation (IITP) Grants funded by the Korean government (MSIT) (No. RS-2024-00437284, No. IITP-2025-RS-2025-02283189 and No. IITP-2025-RS-2025-02263264). This work was supported by Global Partnership Program of Leading Universities in Quantum Science and Technology (RS-2025-08542968) through the National Research Foundation of Korea~(NRF) funded by the Korean government (Ministry of Science and ICT(MSIT)). MO and ORS acknowledge the support from National Science Center, Poland within the QuantERA III Programme (No 2023/05/Y/ST2/00140 acronym Tuquan). This work was also supported in part by the Horizon Europe project Quantum Excellence Centre for Quantum-Enhanced Applications (QEC4QEA).
    The C4QEC project is carried out within the IRAP of the Foundation for Polish Science co-financed by the European Union.
ZZ acknowledges support from the Finnish Quantum Flagship.
\end{acknowledgements}

\bibliography{reference.bib}

@article{oszmaniec2017universal,
  title={Universal extensions of restricted classes of quantum operations},
  author={Oszmaniec, Micha{\l} and Zimbor{\'a}s, Zolt{\'a}n},
  journal={arXiv preprint arXiv:1705.11188},
  year={2017}
}

@book{berezin1966method,
  title={The Method of Second Quantization},
  author={Berezin, Felix A.},
  year={1966},
  publisher={Academic Press}
}

@article{bravyi2005lagrangian,
  title={Lagrangian representation for fermionic linear optics},
  author={Bravyi, Sergey},
  journal={Quantum Information \& Computation},
  volume={5},
  number={3},
  pages={216--238},
  year={2005}
}

@article{hebenstreit2020computational,
  title={Computational power of matchgates with supplementary resources},
  author={Hebenstreit, Martin and Jozsa, Richard and Kraus, Barbara and Strelchuk, Sergii},
  journal={Physical Review A},
  volume={102},
  number={5},
  pages={052604},
  year={2020},
  publisher={APS}
}

@article{ikenmeyer2019hyperpfaffians,
  title={Hyperpfaffians and Geometric Complexity Theory},
  author={Ikenmeyer, Christian and Walter, Michael},
  journal={arXiv preprint arXiv:1912.09389},
  year={2019}
}

@article{baek2023say,
  title={Say no to optimization: A nonorthogonal quantum eigensolver},
  author={Baek, Unpil and Hait, Diptarka and Shee, James and Leimkuhler, Oskar and Huggins, William J and Stetina, Torin F and Head-Gordon, Martin and Whaley, K Birgitta},
  journal={PRX Quantum},
  volume={4},
  number={3},
  pages={030307},
  year={2023},
  publisher={APS}
}

@article{sundstrom2014non,
  title={Non-orthogonal configuration interaction for the calculation of multielectron excited states},
  author={Sundstrom, Eric J and Head-Gordon, Martin},
  journal={The Journal of chemical physics},
  volume={140},
  number={11},
  year={2014},
  publisher={AIP Publishing}
}

@article{zhao2024pt2_oo_upccd,
  title={Enhancing the electron pair approximation with measurements on trapped-ion quantum computers},
  author={Zhao, Luning and Wang, Qingfeng and Goings, Joshua J and Shin, Kyujin and Kyoung, Woomin and Noh, Seunghyo and Rhee, Young Min and Kim, Kyungmin},
  journal={npj Quantum Information},
  volume={10},
  number={1},
  pages={76},
  year={2024},
  publisher={Nature Publishing Group UK London}
}

@article{zhao2023oo_upccd,
  title={Orbital-optimized pair-correlated electron simulations on trapped-ion quantum computers},
  author={Zhao, Luning and Goings, Joshua and Shin, Kyujin and Kyoung, Woomin and Fuks, Johanna I and Kevin Rhee, June-Koo and Rhee, Young Min and Wright, Kenneth and Nguyen, Jason and Kim, Jungsang and others},
  journal={npj Quantum Information},
  volume={9},
  number={1},
  pages={60},
  year={2023},
  publisher={Nature Publishing Group UK London}
}

@article{tecmer2022geminalreview,
  title={Geminal-based electronic structure methods in quantum chemistry. Toward a geminal model chemistry},
  author={Tecmer, Pawe{\l} and Boguslawski, Katharina},
  journal={Physical Chemistry Chemical Physics},
  volume={24},
  number={38},
  pages={23026--23048},
  year={2022},
  publisher={Royal Society of Chemistry}
}

@article{surjan2012sog,
  title={Strongly orthogonal geminals: size-extensive and variational reference states},
  author={Surj{\'a}n, P{\'e}ter R and Szabados, {\'A}gnes and Jeszenszki, P{\'e}ter and Zoboki, Tam{\'a}s},
  journal={Journal of Mathematical Chemistry},
  volume={50},
  number={3},
  pages={534--551},
  year={2012},
  publisher={Springer}
}

@article{stanton1993equation,
  title={The equation of motion coupled-cluster method. A systematic biorthogonal approach to molecular excitation energies, transition probabilities, and excited state properties},
  author={Stanton, John F and Bartlett, Rodney J},
  journal={The Journal of chemical physics},
  volume={98},
  number={9},
  pages={7029--7039},
  year={1993},
  publisher={AIP Publishing}
}

@article{mcclean2017hybrid,
  title={Hybrid quantum-classical hierarchy for mitigation of decoherence and determination of excited states},
  author={McClean, Jarrod R and Kimchi-Schwartz, Mollie E and Carter, Jonathan and De Jong, Wibe A},
  journal={Physical Review A},
  volume={95},
  number={4},
  pages={042308},
  year={2017},
  publisher={APS}
}

@article{ikai2011theory,
  title={On the theory of Pfaffians based on exponential maps in exterior algebras},
  author={Ikai, Hisatoshi},
  journal={Linear algebra and its applications},
  volume={434},
  number={4},
  pages={1094--1106},
  year={2011},
  publisher={Elsevier}
}

@article{brod2011extending,
  title={Extending matchgates into universal quantum computation},
  author={Brod, Daniel J and Galvao, Ernesto F},
  journal={Physical Review A—Atomic, Molecular, and Optical Physics},
  volume={84},
  number={2},
  pages={022310},
  year={2011},
  publisher={APS}
}

@article{mazziotti2012two,
  title = {Two-electron reduced density matrix as the basic variable in many-electron quantum chemistry and physics},
  author = {Mazziotti, David A.},
  journal = {Chemical Reviews},
  volume = {112},
  number = {1},
  pages = {244--262},
  year = {2012}
}

@article{ivanov2016computational,
  title={Computational complexity of exterior products and multi-particle amplitudes of non-interacting fermions in entangled states},
  author={Ivanov, Dmitri A},
  journal={arXiv preprint arXiv:1603.02724},
  year={2016}
}

@article{motta2019efficient,
  title={Efficient ab initio auxiliary-field quantum Monte Carlo calculations in Gaussian bases via low-rank tensor decomposition},
  author={Motta, Mario and Shee, James and Zhang, Shiwei and Chan, Garnet Kin-Lic},
  journal={Journal of chemical theory and computation},
  volume={15},
  number={6},
  pages={3510--3521},
  year={2019},
  publisher={ACS Publications}
}

@article{mcardle2020quantum,
  title={Quantum computational chemistry},
  author={McArdle, Sam and Endo, Suguru and Aspuru-Guzik, Al{\'a}n and Benjamin, Simon C and Yuan, Xiao},
  journal={Reviews of Modern Physics},
  volume={92},
  number={1},
  pages={015003},
  year={2020},
  publisher={APS}
}

@book{helgaker2013molecular,
  title={Molecular electronic-structure theory},
  author={Helgaker, Trygve and Jorgensen, Poul and Olsen, Jeppe},
  year={2013},
  publisher={John Wiley \& Sons}
}

@article{hirsch1983discrete,
  title={Discrete Hubbard-Stratonovich transformation for fermion lattice models},
  author={Hirsch, Jorge E},
  journal={Physical Review B},
  volume={28},
  number={7},
  pages={4059},
  year={1983},
  publisher={APS}
}

@article{d2025majorana,
  title={Majorana string simulation of nonequilibrium dynamics in two-dimensional lattice fermion systems},
  author={D'Anna, Matteo and Nys, Jannes and Carrasquilla, Juan},
  journal={arXiv preprint arXiv:2511.02809},
  year={2025}
}

@article{miller2025simulation,
  title={Simulation of fermionic circuits using majorana propagation},
  author={Miller, Aaron and Favre, Joachim and Holmes, Zo{\"e} and Salehi, {\"O}zlem and Chakraborty, Rahul and Nyk{\"a}nen, Anton and Zimbor{\'a}s, Zolt{\'a}n and Glos, Adam and Garc{\'\i}a-P{\'e}rez, Guillermo},
  journal={arXiv preprint arXiv:2503.18939},
  year={2025}
}

@article{huggins2022unbiasing,
  title={Unbiasing fermionic quantum Monte Carlo with a quantum computer},
  author={Huggins, William J and O’Gorman, Bryan A and Rubin, Nicholas C and Reichman, David R and Babbush, Ryan and Lee, Joonho},
  journal={Nature},
  volume={603},
  number={7901},
  pages={416--420},
  year={2022},
  publisher={Nature Publishing Group UK London}
}

@misc{cudby2024gaussiandecompositionmagicstates,
      title={Gaussian decomposition of magic states for matchgate computations}, 
      author={Joshua Cudby and Sergii Strelchuk},
      year={2024},
      eprint={2307.12654},
      archivePrefix={arXiv},
      primaryClass={quant-ph}
}

@article{ReardonSmith2024improvedsimulation,
  title = {Improved simulation of quantum circuits dominated by free fermionic operations},
  author = {Reardon-Smith, Oliver and Oszmaniec, Micha{\l{}} and Korzekwa, Kamil},
  journal = {{Quantum}},
  issn = {2521-327X},
  publisher = {{Verein zur F{\"{o}}rderung des Open Access Publizierens in den Quantenwissenschaften}},
  volume = {8},
  pages = {1549},
  month = dec,
  year = {2024}
}

@article{Dias_2024,
   title={Classical simulation of non-Gaussian fermionic circuits},
   volume={8},
   ISSN={2521-327X},
   journal={Quantum},
   publisher={Verein zur Forderung des Open Access Publizierens in den Quantenwissenschaften},
   author={Dias, Beatriz and Koenig, Robert},
   year={2024},
   month=may, pages={1350} }

@article{alam2025fermionic,
  title={Fermionic dynamics on a trapped-ion quantum computer beyond exact classical simulation},
  author={Alam, Faisal and Bosse, Jan Lukas and {\v{C}}epait{\.e}, Ieva and Chapman, Adrian and Clinton, Laura and Crichigno, Marcos and Crosson, Elizabeth and Cubitt, Toby and Derby, Charles and Dowinton, Oliver and others},
  journal={arXiv preprint arXiv:2510.26300},
  year={2025}
}

@article{google2020hartree,
  title={Hartree-Fock on a superconducting qubit quantum computer},
  author={Arute, Frank and Arya, Kunal and Babbush, Ryan and Bacon, Dave and Bardin, Joseph C and Barends, Rami and Boixo, Sergio and Broughton, Michael and Buckley, Bob B and others},
  journal={Science},
  volume={369},
  number={6507},
  pages={1084--1089},
  year={2020},
  publisher={American Association for the Advancement of Science}
}

@article{huggins2020non,
  title={A non-orthogonal variational quantum eigensolver},
  author={Huggins, William J and Lee, Joonho and Baek, Unpil and O’Gorman, Bryan and Whaley, K Birgitta},
  journal={New Journal of Physics},
  volume={22},
  number={7},
  pages={073009},
  year={2020},
  publisher={IOP Publishing}
}

@article{ollitrault2024enhancing,
  title={Enhancing initial state overlap through orbital optimization for faster molecular electronic ground-state energy estimation},
  author={Ollitrault, Pauline J and Cortes, Cristian L and Gonthier, J{\'e}r{\^o}me F and Parrish, Robert M and Rocca, Dario and Anselmetti, Gian-Luca and Degroote, Matthias and Moll, Nikolaj and Santagati, Raffaele and Streif, Michael},
  journal={Physical Review Letters},
  volume={133},
  number={25},
  pages={250601},
  year={2024},
  publisher={APS}
}

@article{daley2022practical,
  title={Practical quantum advantage in quantum simulation},
  author={Daley, Andrew J and Bloch, Immanuel and Kokail, Christian and Flannigan, Stuart and Pearson, Natalie and Troyer, Matthias and Zoller, Peter},
  journal={Nature},
  volume={607},
  number={7920},
  pages={667--676},
  year={2022},
  publisher={Nature Publishing Group UK London}
}

@article{arute2020observation,
  title={Observation of separated dynamics of charge and spin in the Fermi-Hubbard model},
  author={Arute, Frank and Arya, Kunal and Babbush, Ryan and Bacon, Dave and Bardin, Joseph C and Barends, Rami and Bengtsson, Andreas and Boixo, Sergio and Broughton, Michael and Buckley, Bob B and others},
  journal={arXiv preprint arXiv:2010.07965},
  year={2020}
}

@article{aaronson2012generalizing,
  title={Generalizing and derandomizing Gurvits's approximation algorithm for the permanent},
  author={Aaronson, Scott and Hance, Travis},
  journal={arXiv preprint arXiv:1212.0025},
  year={2012}
}

@article{leimkuhler2025exponential,
  title={Exponential quantum speedups in quantum chemistry with linear depth},
  author={Leimkuhler, Oskar and Whaley, K Birgitta},
  journal={arXiv preprint arXiv:2503.21041},
  year={2025}
}

@article{oszmaniec2022fermion,
  title={Fermion sampling: a robust quantum computational advantage scheme using fermionic linear optics and magic input states},
  author={Oszmaniec, Micha{\l} and Dangniam, Ninnat and Morales, Mauro ES and Zimbor{\'a}s, Zolt{\'a}n},
  journal={PRX Quantum},
  volume={3},
  number={2},
  pages={020328},
  year={2022},
  publisher={APS}
}

@article{recio2025train,
  title={Train on classical, deploy on quantum: scaling generative quantum machine learning to a thousand qubits},
  author={Recio-Armengol, Erik and Ahmed, Shahnawaz and Bowles, Joseph},
  journal={arXiv preprint arXiv:2503.02934},
  year={2025}
}

@article{rudolph2024trainability,
  title={Trainability barriers and opportunities in quantum generative modeling},
  author={Rudolph, Manuel S and Lerch, Sacha and Thanasilp, Supanut and Kiss, Oriel and Shaya, Oxana and Vallecorsa, Sofia and Grossi, Michele and Holmes, Zo{\"e}},
  journal={npj Quantum Information},
  volume={10},
  number={1},
  pages={116},
  year={2024},
  publisher={Nature Publishing Group UK London}
}

@article{bako2025fermionic,
  title={Fermionic Born Machines: Classical training of quantum generative models based on Fermion Sampling},
  author={Bak{\'o}, Bence and Kolarovszki, Zolt{\'a}n and Zimbor{\'a}s, Zolt{\'a}n},
  journal={arXiv preprint arXiv:2511.13844},
  year={2025}
}

@article{lim2025efficient,
  title={Efficient classical algorithms for linear optical circuits},
  author={Lim, Youngrong and Oh, Changhun},
  journal={arXiv preprint arXiv:2502.12882},
  year={2025}
}

@article{seron2024efficient,
  title={Efficient validation of boson sampling from binned photon-number distributions},
  author={Seron, Benoit and Novo, Leonardo and Arkhipov, Alex and Cerf, Nicolas J},
  journal={Quantum},
  volume={8},
  pages={1479},
  year={2024},
  publisher={Verein zur F{\"o}rderung des Open Access Publizierens in den Quantenwissenschaften}
}

@article{oh2024quantum,
  title={Quantum-inspired classical algorithms for molecular vibronic spectra},
  author={Oh, Changhun and Lim, Youngrong and Wong, Yat and Fefferman, Bill and Jiang, Liang},
  journal={Nature Physics},
  volume={20},
  number={2},
  pages={225--231},
  year={2024},
  publisher={Nature Publishing Group UK London}
}

@article{terhal2002classical,
  title={Classical simulation of noninteracting-fermion quantum circuits},
  author={Terhal, Barbara M and DiVincenzo, David P},
  journal={Physical Review A},
  volume={65},
  number={3},
  pages={032325},
  year={2002},
  publisher={APS}
}

@article{wimmer2012algorithm,
  title={Algorithm 923: Efficient numerical computation of the pfaffian for dense and banded skew-symmetric matrices},
  author={Wimmer, Michael},
  journal={ACM Transactions on Mathematical Software (TOMS)},
  volume={38},
  number={4},
  pages={1--17},
  year={2012},
  publisher={ACM New York, NY, USA}
}

@article{dias2024classical,
  title={Classical simulation of non-Gaussian fermionic circuits},
  author={Dias, Beatriz and Koenig, Robert},
  journal={Quantum},
  volume={8},
  pages={1350},
  year={2024},
  publisher={Verein zur F{\"o}rderung des Open Access Publizierens in den Quantenwissenschaften}
}

@article{knill2001fermionic,
  title={Fermionic linear optics and matchgates},
  author={Knill, Emanuel},
  journal={arXiv preprint quant-ph/0108033},
  year={2001}
}

@article{valiant2002quantum,
  title={Quantum circuits that can be simulated classically in polynomial time},
  author={Valiant, Leslie G},
  journal={SIAM Journal on Computing},
  volume={31},
  number={4},
  pages={1229--1254},
  year={2002},
  publisher={SIAM}
}

@book{fetter2012quantum,
  title={Quantum theory of many-particle systems},
  author={Fetter, Alexander L and Walecka, John Dirk},
  year={2012},
  publisher={Courier Corporation}
}

@article{knuth1995overlapping,
  title={Overlapping pfaffians},
  author={Knuth, Donald E},
  journal={arXiv preprint math/9503234},
  year={1995}
}

@article{jozsa2008matchgates,
  title={Matchgates and classical simulation of quantum circuits},
  author={Jozsa, Richard and Miyake, Akimasa},
  journal={Proceedings of the Royal Society A: Mathematical, Physical and Engineering Sciences},
  volume={464},
  number={2100},
  pages={3089--3106},
  year={2008},
  publisher={The Royal Society London}
}

@inproceedings{gurvits2005complexity,
  title={On the complexity of mixed discriminants and related problems},
  author={Gurvits, Leonid},
  booktitle={International Symposium on Mathematical Foundations of Computer Science},
  pages={447--458},
  year={2005},
  organization={Springer}
}

@book{horn-and-johnson,
    doi = {10.1017/CBO9781139020411},
    title = {Matrix Analysis},
    authors = {Horn, Roger A. and Johnson, Charles R.},
    year = {2012},
    edition = {2},
    publisher = {Cambridge University Press},
}

@inproceedings{aaronson2011computational,
  title={The computational complexity of linear optics},
  author={Aaronson, S. and Arkhipov, A.},
  booktitle={Proceedings of the forty-third annual ACM symposium on Theory of computing},
  pages={333--342},
  year={2011}
}

@article{valiant1979complexity,
  title={The complexity of computing the permanent},
  author={Valiant, Leslie G},
  journal={Theoretical computer science},
  volume={8},
  number={2},
  pages={189--201},
  year={1979},
  publisher={Elsevier}
}

\appendix

\section{Proof of the Pfaffian-Wedge Identity}\label{app:pf_identity_proof}
In this Appendix, we provide an explicit derivation of the fundamental algebraic identity relating the Pfaffian of a skew-symmetric matrix $A \in \mathbb{C}^{2N \times 2N}$ to the $N$-th exterior power of its associated two-form for completeness. 
This geometric correspondence formally underpins the exterior-algebra coefficient-extraction framework utilized throughout the main text.

Let $V \cong \mathbb{C}^{2N}$ be a vector space equipped with a canonical basis $\{e_1, \dots, e_{2N}\}$. The associated two-form for the matrix $A$ is defined via the linear map:
\begin{equation}
\alpha(A) \equiv \frac{1}{2} \sum_{p,q=1}^{2N} A_{pq} e_p \wedge e_q.
\end{equation}
Evaluating the $N$-th exterior power (the wedge product of $N$ identical two-forms) yields the macroscopic multiparticle state:
\begin{align} \label{eq:wedge_expansion}
\frac{1}{N!} \alpha(A)^{\wedge N} &= \frac{1}{2^N N!} \sum_{p_1, q_1, \dots, p_N, q_N = 1}^{2N} \left( \prod_{k=1}^N A_{p_k q_k} \right) \nonumber \\ 
&\times (e_{p_1} \wedge e_{q_1}) \wedge \dots \wedge (e_{p_N} \wedge e_{q_N}).
\end{align}
Due to the total antisymmetry of the exterior product ($e_i \wedge e_j = -e_j \wedge e_i$ and $e_i \wedge e_i = 0$), any term containing repeated basis vectors vanishes. Consequently, non-zero contributions occur if and only if the sequence of indices $(p_1, q_1, \dots, p_N, q_N)$ forms a valid permutation $\sigma \in \mathcal{S}_{2N}$ of the integer set $\{1, 2, \dots, 2N\}$.

Rearranging these basis vectors into the canonical ascending lexicographical order $e_1 \wedge \dots \wedge e_{2N}$ naturally introduces the sign of the permutation, $\mathrm{sgn}(\sigma) \in \{\pm 1\}$, exactly reproducing the physical anticommutation exchange signs of the fermions. 
Rewriting Eq.~\eqref{eq:wedge_expansion} as a sum over the symmetric group $\mathcal{S}_{2N}$ gives:
\begin{align}
    &\frac{1}{N!} \alpha(A)^{\wedge N}  \\
    &= \left( \frac{1}{2^N N!} \sum_{\sigma \in \mathcal{S}_{2N}} \mathrm{sgn}(\sigma) \prod_{i=1}^N A_{\sigma(2i-1), \sigma(2i)} \right) \nonumber \\ 
    &~~~\times e_1 \wedge \dots \wedge e_{2N} \\
    &= \mathrm{pf}(A)(e_1 \wedge \dots \wedge e_{2N}).
\end{align}
By defining the canonical volume form as $\mathrm{vol} \equiv e_1 \wedge \dots \wedge e_{2N}$, and applying the linear functional $[\cdot]_{\mathrm{vol}}$ that extracts the scalar coefficient of this top-form, we recover the exact defining identity:
\begin{equation}
    \mathrm{pf}(A) = \frac{1}{N!} \left[ \alpha(A)^{\wedge N} \right]_{\mathrm{vol}}.
\end{equation}

\section{Algebraic construction and variance bounds for APSG estimators}\label{app:general_estimator}
This Appendix collects the fixed-output coefficient identity, the randomized sign filter, and the pointwise Pfaffian bounds used in the proofs for block APSG states.

Let $|\Psi\rangle=\prod_{t=1}^N\hat\eta_t^\dagger|0\rangle$ be a block-product APSG state as in Def.~\ref{def:disjointBGP}, and let $W_t$ be the
skew-symmetric matrix representing the pair creator $\hat\eta_t^\dagger$ on block $t$. Let $\hat G$ be a number-preserving Gaussian map with single-particle matrix $G$. For a fixed $2N$-particle output basis state $|x\rangle$, with occupied-mode set $S_x$, define
\begin{align}
    X=(G_{S_x,:})^\T,\qquad B_t(x)=X^\T W_t X .
\end{align}
For a passive FLO unitary, $G=U$ and $X$ is an isometry.

Define
\begin{align}
    P_x(y_1,\dots,y_N)
    =
    \pf\!\left(\sum_{t=1}^N y_t B_t(x)\right).
\end{align}
Then the fixed-output amplitude is the multilinear coefficient
\begin{align}
    \langle x|\hat G|\Psi\rangle
    =
    [y_1\cdots y_N]P_x(y_1,\dots,y_N).
\end{align}
Indeed, the Pfaffian--wedge identity converts the top-form coefficient of
$\sum_t y_t\alpha(B_t(x))$ into the Pfaffian above, while extracting
$y_1\cdots y_N$ enforces the APSG rule that exactly one pair is selected from
each block.

Let $b=(b_1,\dots,b_N)\sim\mathrm{Unif}(\{\pm1\}^N)$ and define
\begin{align}
    \mathcal Z(b;x)
    =
    \pf\!\left(\sum_{t=1}^N b_tB_t(x)\right)
    \prod_{t=1}^N b_t .
\end{align}

\begin{lemma}[Exact coefficient isolation]\label{lem:sign_filter}
For every fixed output $x$,
\begin{align}
    \E_b[\mathcal Z(b;x)]
    =
    [y_1\cdots y_N]P_x(y).
\end{align}
\end{lemma}

\begin{proof}
Write $P_x(y)=\sum_a c_a\prod_t y_t^{a_t}$. Since the Pfaffian is homogeneous
of degree $N$, only monomials with $\sum_t a_t=N$ appear. Multiplication by
$\prod_t b_t$ and averaging over independent signs kills every monomial for
which some $a_t+1$ is odd. The only surviving monomial has $a_t=1$ for all
$t$, namely $y_1\cdots y_N$.
\end{proof}

\begin{lemma}[Pointwise Pfaffian bound]\label{lem:pf_bound_apsg}
Let $W_1,\ldots,W_N$ be supported on mutually disjoint mode blocks and satisfy
$\|W_t\|_{\mathrm{op}}\le1$. If $\|X\|_{\mathrm{op}}\le L$, then for every
$b\in\{\pm1\}^N$,
\begin{align}
    \left|
    \pf\!\left[X^\T\left(\sum_{t=1}^N b_tW_t\right)X\right]
    \right|
    \le L^{2N}.
\end{align}
In particular, the bound is $1$ for passive FLO unitaries and
$\|G\|_{\mathrm{op}}^{2N}$ for a number-preserving Gaussian map with
single-particle matrix $G$.
\end{lemma}

\begin{proof}
Set $W(b)=\sum_t b_tW_t$. Since the block supports are disjoint,
$\|W(b)\|_{\mathrm{op}}\le1$. Hence
\begin{align}
    \|X^\T W(b)X\|_{\mathrm{op}}\le \|X\|_{\mathrm{op}}^2\|W(b)\|_{\mathrm{op}}
    \le L^2.
\end{align}
For any $2N\times2N$ skew-symmetric matrix $A$,
$|\pf(A)|^2=|\det A|\le \|A\|_{\mathrm{op}}^{2N}$, so
$|\pf(A)|\le \|A\|_{\mathrm{op}}^N\le L^{2N}$.
For $G=U$ unitary, $X=(U_{S_x,:})^\T$ is an isometry and $L=1$; in general
$L\le\|G\|_{\mathrm{op}}$.
\end{proof}

Consequently, each fixed-output sample is bounded by $1$ in the unitary case and by $\|G\|_{\mathrm{op}}^{2N}$ in the non-unitary case. Each sample requires one Pfaffian of size $2N\times2N$, which can be evaluated in $O(N^3)$ time.

\section{Proofs of the APSG overlap estimators}
\label{app:apsg_overlap_proofs}

We first note that Theorem~\ref{thm:overlap_fock_apsg} follows directly from Lemma~\ref{lem:pf_bound_apsg} by writing
\begin{align}
    \ket{\Psi} &= \bigotimes_{t=1}^N\left(\sum_{j=1}^{r_t} w_{t,j}\,\hat p_{t,j}^\dagger\ket{0}\right)\\
    &=\left(\prod_{t=1}^N \gamma_t \right) \bigotimes_{t=1}^N\left(\sum_{j=1}^{r_t} \frac{w_{t,j}}{\gamma_t}\,\hat p_{t,j}^\dagger\ket{0}\right),
\end{align}
where we have defined $\gamma_t = \max_{j\in [r_t]} w_{t,j}$, and noting that the block-matrices $W_t$ have operator norm $1$ after dividing each element by $\gamma_t$.

We now prove Theorems~\ref{thm:overlap_apsg} and~\ref{thm:overlap_apsg_nonunitary}. 
As noted in the main text, the estimator used is directly adapted from that of Theorem~\ref{thm:overlap_fock_apsg}. The proof is identical in both cases, with the only difference being that in the case where the Gaussian map $\hat{G}$ is unitary, it satisfies $\norm{G}_\infty = 1$. 
We rewrite 
\begin{align}
    \langle \Phi|\hat{G}|\Psi\rangle  &= \bigotimes_{t=1}^N\left(\sum_{j=1}^{s_t} v_{t,j}^*\,\bra{0}\hat{q}_{t,j}\right)\hat{G}|\Psi\rangle\\
    &= \bigotimes_{t=1}^N\left(\sum_{j=1}^{s_t} \abs{v_{t,j}}^2\, \frac{\bra{0}\hat{q}_{t,j}}{v_{t,j}}\right)\hat{G}|\Psi\rangle,
\end{align}
from which it follows that if we sample $x$ from the probability distribution $P(x) = \prod_{t=1}^N \abs{v_{t, x_t}}^2$, the overlap we seek is exactly given by the expectation value
\begin{align}
    \langle \Phi|\hat{G}|\Psi\rangle &= \mathop{\mathbb{E}}_{x\sim P} \left[\left(\bigotimes_{t=1}^N\frac{\bra{0}\hat{q}_{t,x_t}}{v_{t,x_t}}\right) \hat{G}|\Psi\rangle\right].\label{eqn:block-aspg-overlap-one-average}
\end{align}
It is convenient to introduce the notation $e_x$ for the binary vector such that $\bigotimes_{t=1}^N \hat{q}_{t,x_t}^\dagger \ket{0} = \ket{e_x}$, we now employ the estimator of Lemma~\ref{lem:pf_bound_apsg} for the inner products $\bra{e_x}\hat{G}\ket{\Psi}$ appearing in Eq.~\eqref{eqn:block-aspg-overlap-one-average} to obtain
\begin{align}
    &\langle \Phi|\hat{G}|\Psi\rangle = \mathop{\mathbb{E}}_{\overset{x\sim P}{b\sim \text{unif}}}  \left[\left(\prod_{t=1}^Nv_{t,x_t}^{-1}\right) \mathcal Z(b;e_x) \right]\\
    &= \mathop{\mathbb{E}}_{\overset{x\sim P}{b\sim \text{unif}}}  \left[\left(\prod_{t=1}^Nv_{t,x_t}^{-1}\right)\pf\!\left(G_{e_x}\sum_{t=1}^N b_t W_t G_{e_x}^T\right)
    \prod_{t=1}^N b_t \right],
\end{align} 
where $G_{e_x}$ denotes the matrix formed by deleting the columns of $G$ for which $e_x$ is $0$. To prove Theorems~\ref{thm:overlap_apsg} and~\ref{thm:overlap_apsg_nonunitary}, all that remains is to show that the variance of this random variable is bounded. The variance takes the form
\begin{align}
    &\operatorname{Var}(Z) \\
    &=\mathop{\mathbb{E}}_{\overset{x\sim P}{b\sim \text{unif}}}  \left[\left(\prod_{t=1}^N|v_{t,x_t}|^{-2}\right) \abs{\pf\!\left(G_{e_x}\left(\sum_{t=1}^N b_t W_t\right) G_{e_x}^T\right)}^2 \right]\\
    &=\mathop{\mathbb{E}}_{b\sim \text{unif}}  \left[\sum_{x\in \prod_{t=1}^N [s_t]}\abs{\pf\!\left(G_{e_x}\left(\sum_{t=1}^N b_t W_t\right) G_{e_x}^T\right)}^2 \right],
\end{align} 
where the probability $P(x)$ has canceled with the product of $|v_{t,x_t}|^{-2}$. 
We show that, in fact, a pointwise bound holds for the quantity that the remaining expectation value averages over.

\begin{lemma}
    If
    \begin{align}
    S &= \sum_{x\in \prod_{t=1}^N [s_t]}\abs{\pf\!\left(G_{e_x} W G_{e_x}^T\right)}^2,
\end{align}
then 
\begin{align}
    S \leq \left(\frac{1}{2N} \norm{G W G^T }_{\mathrm{HS}}^2\right)^N,
\end{align}
in particular if $W$ is obtained from a normalized APSG state $\ket{\Psi}$ so
\begin{align}
    W &= \sum_{t=1}^N b_t W_t = \bigoplus_{t=1}^N b_t\bigoplus_{j=1}^{r_t}
    \begin{pmatrix}
        0 & w_{t,j}\\-w_{t,j} & 0
    \end{pmatrix}\\
    \norm{W}_{\mathrm{HS}}^2 &= 2N,
\end{align}
and we have $S\leq \norm{G}_{\infty}^{4N}$.
\end{lemma}

\begin{proof}
We introduce the matrices
\begin{align}
    A &= G W G^T\\
    B &= A^\dagger A 
    = G^* W^\dagger G^\dagger G W G^T,
\end{align}
the subsequent derivation will rely on the fact that $B$ is positive semi-definite. 
We recall that the principal submatrices $B_{e_x, e_x}$ of $B$ are also positive semidefinite and apply the Cauchy-Binet formula for determinants~\cite{horn-and-johnson} to the submatrix $B_{e_x, e_x}$ to obtain
\begin{align}
    \det(B_{e_x,e_x}) &= \det(A_{e_x}^\dagger A_{e_x})\\
    &= \sum_{s\in S} \det(A_{e_x,s}^\dagger)\det(A_{e_x,s})\\
    &= \sum_{s\in S} \abs{\det(A_{e_x,s})}^2,
\end{align}
where $S$ consists of subsets of indices of size $2N$. 
In particular, $e_x\in S$, so we obtain the inequality
\begin{align}
    \abs{\pf(A_{e_x, e_x})}^2 = \abs{\det(A_{e_x,e_x})} \leq \det(B_{e_x,e_x})^{1/2}.
\end{align}
We now apply Fischer's inequality~\cite{horn-and-johnson} with $2\times 2$ blocks to the submatrix determinants 
\begin{align}
    \det(B_{e_x,e_x}) \leq \prod_{j=1}^N \det(B_{e_x,e_x}[j]),
\end{align}
where $B_{e_x,e_x}[j]$ denotes the $j^\text{th}$ $2\times 2$ submatrix appearing on the (block) diagonal of $B_{e_x, e_x}$. Combining these inequalities, we obtain
\begin{align}
S &= \sum_{x\in \prod_{t=1}^N [s_t]}\abs{\pf\!\left(G_{e_x} W G_{e_x}^T\right)}^2\\
&\leq \sum_{x\in \prod_{t=1}^N [s_t]} \prod_{j=1}^{N} \det(B_{e_x,e_x}[j])^{1/2}.\label{eqn:app-s-with-inequalities}
\end{align}
Recall that given a choice of $x$, $e_x$ picks out a single pair of adjacent indices for each tensor factor in the original APSG state $\ket{\Psi}$. The sum in Eq.~\eqref{eqn:app-s-with-inequalities} sums over every possible way of choosing such adjacent pairs, so we have
\begin{align}
    \sum_{x\in \prod_{t=1}^N [s_t]} \prod_{j=1}^{N} \det(B_{e_x,e_x}[j])^{1/2} &= \prod_{t=1}^N\sum_{x\in [s_t]}\det(B_{e_x,e_x}[t])^{1/2}.
\end{align}
The result follows from two further inequalities, firstly, for $2\times 2$ positive semi-definite matrices, we have
\begin{align}
    \det(B_{e_x,e_x}[t])^{1/2} \leq \frac{1}{2} \tr(B_{e_x,e_x}[t]),
\end{align}
and secondly, the arithmetic-geometric mean inequality. Applying these, we have
\begin{align}
    S &\leq\prod_{t=1}^N\sum_{x\in [s_t]}\frac{1}{2}\tr(B_{e_x,e_x}[t])\\ 
    &\leq \left(\frac{1}{2N} \sum_{t=1}^N \sum_{x\in [s_t]} \tr(B_{e_x, e_x}[t])\right)^N\\
    &= \left(\frac{1}{2N} \tr(B)\right)^N.
\end{align}
Since $B= A^\dagger A$, we have $\tr(B) = \norm{A}_{\mathrm{HS}}^2$.
\end{proof}

\section{Trapped-ion fermionic dynamics benchmark: definitions and conventions}
\label{app:phasecraft_model}

This Appendix provides the definitions and normalization conventions used for the trapped-ion fermionic-dynamics benchmark of Ref.~\cite{alam2025fermionic}, in a self-contained form.

Ref.~\cite{alam2025fermionic} considers the spinful single-band Fermi--Hubbard Hamiltonian
\begin{equation}
    \hat H
    =-J\sum_{\langle i,j\rangle,\sigma}\Big(e^{i\phi_{ij}}\hat c_{i\sigma}^\dagger \hat c_{j\sigma}+\mathrm{h.c.}\Big)
    +W\sum_i \hat n_{i\uparrow}\hat n_{i\downarrow},
\end{equation}
where $\hat n_{i\sigma}=\hat c_{i\sigma}^\dagger \hat c_{i\sigma}$ and $\langle i,j\rangle$ denotes nearest-neighbor pairs on the lattice. The real-time state is $\ket{\Psi(t)}=e^{-it\hat H}\ket{\Psi_0}$. In the noninteracting limit $W=0$, $\hat H$ is quadratic and number preserving, so the evolution $\hat U(t)=e^{-it\hat H}$ is a passive FLO.

A key ingredient of Ref.~\cite{alam2025fermionic} is a structured, non-Gaussian initial state built from a fixed nearest-neighbor dimer covering, together with localized holon/doublon defects. On a dimer $(j,k)$, the $S^z_{\mathrm{tot}}=0$ triplet is
\begin{equation}
|T_0\rangle_{j,k}
=
\frac{1}{\sqrt 2}\Big(\hat c_{j\uparrow}^\dagger \hat c_{k\downarrow}^\dagger
+\hat c_{j\downarrow}^\dagger \hat c_{k\uparrow}^\dagger\Big)|0\rangle.
\end{equation}
The global input $\ket{\Psi_0}$ is a tensor product of such triplets over the chosen dimer covering, with one empty site (holon) and one doubly occupied site (doublon) inserted at specified locations, following Ref.~\cite{alam2025fermionic}.

Among the reported diagnostics, we use local densities $n_{i\sigma}(t)=\langle \hat n_{i\sigma}(t)\rangle$, the connected spin correlator
\begin{equation}
C_{zz}(i,j;t)=4\Big(\langle \hat S_i^z(t)\hat S_j^z(t)\rangle-\langle \hat S_i^z(t)\rangle\langle \hat S_j^z(t)\rangle\Big),
\end{equation}
with $\hat S_i^z=\tfrac12(\hat n_{i\uparrow}-\hat n_{i\downarrow})$, and the total doublon number $\hat N_{\mathrm{doublon}}=\sum_i \hat n_{i\uparrow}\hat n_{i\downarrow}$. The experiment also reports a nearest-neighbor triplet diagnostic defined as a spatial average of $C_{zz}(i,j;t)$ over all nearest-neighbor links; in Ref.~\cite{alam2025fermionic} this average is normalized by the 2D system size $L_xL_y$ (for $L=L_xL_y$ sites).

Finally, Ref.~\cite{alam2025fermionic} introduces high-weight diagonal string observables (Wilson lines/loops) constructed from local holon/doublon projectors along paths or loops. Writing the doublon projector as $\hat d_i\equiv \hat n_{i\uparrow}\hat n_{i\downarrow}$, a representative Wilson-loop-type observable is a product of doublon-free projectors along a closed curve $C$,
\begin{equation}
\hat W_C \equiv \prod_{i\in C} (1-\hat d_i),
\end{equation}
and the open Wilson line $V_{hd}$ is defined analogously along an open path connecting the holon and doublon defects (see Ref.~\cite{alam2025fermionic} for the precise path conventions used in the experiment). In all cases, these diagnostics are diagonal polynomials in $\{\hat n_{i\sigma}\}$ and $\{1-\hat n_{i\uparrow}\hat n_{i\downarrow}\}$, and therefore admit parity-string/Walsh--Hadamard expansions in general. For the doublon Wilson loops used in the numerical benchmark of Sec.~\ref{sec:phasecraft}, however, we employ the more specialized charge-phase estimator described below in this Appendix.

\paragraph*{Charge-phase estimator for doublon Wilson loops.}
For the numerical Wilson-loop benchmark of Sec.~\ref{sec:phasecraft}, we
specialize the general diagonal-observable framework to the doublon-free string
\begin{equation}
    \hat W_C=\prod_{j\in C}(1-\hat d_j),
    \qquad
    \hat d_j\equiv \hat n_{j\uparrow}\hat n_{j\downarrow}.
\end{equation}
Introduce the local charge operator
\begin{equation}
    \hat q_j\equiv \hat n_{j\uparrow}+\hat n_{j\downarrow}.
\end{equation}
Since $1-\hat d_j$ depends only on the three-valued charge
$q_j\in\{0,1,2\}$, one has the exact identity
\begin{align}
    1-\hat d_j
    &=
    \frac{1}{\sqrt{3}}e^{-i\pi/6}e^{\,i(\pi/3)\hat q_j}
    +
    \frac{1}{\sqrt{3}}e^{\,i\pi/6}e^{-\,i(\pi/3)\hat q_j} \\
    \label{eq:wilson_local_charge_phase}
    &=
    \frac{2}{\sqrt{3}}
    \mathbb{E}_{\sigma_j\sim\mathrm{Unif}\{\pm1\}}
    \left[
        e^{-i\sigma_j\pi/6}
        e^{\,i\sigma_j(\pi/3)\hat q_j}
    \right].
\end{align}
Multiplying over all sites of the contour $C$ gives
\begin{equation}
    \hat W_C
    =
    \left(\frac{2}{\sqrt{3}}\right)^{|C|}
    \mathbb{E}_{\boldsymbol{\sigma}\sim\mathrm{Unif}(\{\pm1\}^{|C|})}
    \left[
        e^{-i\frac{\pi}{6}\sum_{j\in C}\sigma_j}
        \hat D_C(\boldsymbol{\sigma})
    \right],
\end{equation}
where
\begin{equation}
    \hat D_C(\boldsymbol{\sigma})
    \equiv
    \exp\!\left(
        i\frac{\pi}{3}\sum_{j\in C}\sigma_j \hat q_j
    \right).
\end{equation}

For the noninteracting evolution $\hat U_0(t)=e^{-it\hat H_0}$, define
\begin{equation}
    \mu_C(t;\boldsymbol{\sigma})
    \equiv
    \bra{\Psi_0}
    \hat U_0^\dagger(t)\hat D_C(\boldsymbol{\sigma})\hat U_0(t)
    \ket{\Psi_0}.
\end{equation}
Because $\hat D_C(\boldsymbol{\sigma})$ is the second-quantized action of a
diagonal passive-FLO unitary, $\mu_C(t;\boldsymbol{\sigma})$ is exactly an
instance of the overlap primitive covered by
Theorem~\ref{thm:overlap_apsg}. Therefore,
\begin{equation}
    \langle \hat W_C(t)\rangle
    =
    \left(\frac{2}{\sqrt{3}}\right)^{|C|}
    \mathbb{E}_{\boldsymbol{\sigma}}
    \left[
        e^{-i\frac{\pi}{6}\sum_{j\in C}\sigma_j}
        \mu_C(t;\boldsymbol{\sigma})
    \right].
\end{equation}

\begin{proposition}
Let $\mathcal Z_C(t;\boldsymbol{\sigma})$ be a one-shot mixed-Pfaffian estimator for
$\mu_C(t;\boldsymbol{\sigma})$. Then the combined one-shot random variable
\begin{equation}
    \mathcal X_C(t;\boldsymbol{\sigma})
    \equiv
    \left(\frac{2}{\sqrt{3}}\right)^{|C|}
    e^{-i\frac{\pi}{6}\sum_{j\in C}\sigma_j}
    \mathcal Z_C(t;\boldsymbol{\sigma})
\end{equation}
is unbiased for $\langle \hat W_C(t)\rangle$ and satisfies
\begin{equation}
    |\mathcal X_C(t;\boldsymbol{\sigma})|
    \le
    \left(\frac{2}{\sqrt{3}}\right)^{|C|}.
\end{equation}
Consequently, averaging $K$ i.i.d.\ samples yields an additive-error estimator
with sample complexity
\begin{equation}
    K
    =
    O\!\left(
        \left(\frac{4}{3}\right)^{|C|}
        \epsilon^{-2}
        \log\delta^{-1}
    \right).
\end{equation}
\end{proposition}

\begin{proof}
By Lemma~\ref{lem:pf_bound_apsg}, the inner mixed-Pfaffian one-shot satisfies $|\mathcal Z_C(t;\boldsymbol{\sigma})|\le 1$. The magnitude bound then follows immediately from the prefactor in the definition of $\mathcal X_C$. Unbiasedness follows from the exact charge-phase decomposition above together with the unbiasedness of the mixed-Pfaffian overlap estimator. The stated sample complexity is then a direct consequence of Hoeffding's inequality.
\end{proof}

This is the estimator used for the Wilson-loop numerics in
Fig.~\ref{fig:noninteracting}.

\section{Interacting dynamics (\texorpdfstring{$W>0$}{W > 0}) and the real-time variance barrier}
\label{app:interacting_dynamics}

In this Appendix, we detail the simulation boundary for the interacting dynamics ($W>0$) discussed in Sec.~\ref{sec:phasecraft}, explicitly demonstrating how the non-unitary nature of the real-time Hubbard--Stratonovich (HS) transformation manifests as an exponential variance barrier.

We fix a step size $\Delta t$ and write $t=n\Delta t$, analyzing the Strang-split Trotterized evolution:
\begin{align}
    e^{-it(\hat H_0+\hat H_W)} \approx
    \left(e^{-i(\Delta t/2)\hat H_0} e^{-i\Delta t\hat H_W} e^{-i(\Delta t/2)\hat H_0}\right)^n,
    \label{eq:strang-fixed-W}
\end{align}
focusing exclusively on the sampling complexity for this Trotterized dynamics (ignoring Trotter error). 

To simulate this interacting regime, we map the time evolution onto an ensemble of non-interacting evolutions using a discrete real-time HS transformation~\cite{hirsch1983discrete}. For each site $i\in[L]$ and each time slice, we use the centered interaction factor:
\begin{align}
    &e^{-i\Delta t W(\hat n_{i\uparrow}-\tfrac12)(\hat n_{i\downarrow}-\tfrac12)} \nonumber \\
    &=
    e^{-i\Delta t W/4}\cdot \frac12\sum_{\sigma_i=\pm1}
    \exp \left(\lambda \sigma_i(\hat n_{i\uparrow}-\hat n_{i\downarrow})\right),
    \label{eq:hirsch-rt-W}
\end{align}
where $\lambda=\lambda(W,\Delta t)$ satisfies $\cosh\lambda = e^{i\Delta t W/2}$. Thus, each time slice $\ell\in\{1,\dots,n\}$ is associated with an auxiliary-field configuration $\boldsymbol{\sigma}_\ell\in\{\pm1\}^L$.

Let $U_{1/2}=e^{-i(\Delta t/2)h_0}$ be the single-particle half-step propagator corresponding to $\hat H_0$. The HS one-body factor is diagonal:
\begin{align}
    V(\boldsymbol{\sigma}_\ell)
    =
    \mathrm{diag}\big(
    e^{-\lambda\sigma_{\ell,1}},\dots,e^{-\lambda\sigma_{\ell,L}},
    e^{+\lambda\sigma_{\ell,1}},\dots,e^{+\lambda\sigma_{\ell,L}}
    \big).
    \label{eq:Vdiag-W}
\end{align}
For a complete HS path $\boldsymbol{\sigma}=(\boldsymbol{\sigma}_1,\dots,\boldsymbol{\sigma}_n)$, the corresponding single-particle propagator is:
\begin{equation}
    G(\boldsymbol{\sigma})
    \equiv
    \prod_{\ell=1}^{n}\Big(U_{1/2} V(\boldsymbol{\sigma}_\ell) U_{1/2}\Big).
    \label{eq:G-sigma-W}
\end{equation}

Under this representation, each parity expectation $\mu_T(t)=\bra{\Psi(t)}\hat{\Pi}_T\ket{\Psi(t)}$ becomes an average over independent forward and backward auxiliary-field paths, $\boldsymbol{\sigma}_L$ and $\boldsymbol{\sigma}_R$:
\begin{equation}
    \mu_T(t)=\mathbb{E}_{\boldsymbol{\sigma}_L,\boldsymbol{\sigma}_R} \left[
    X_T(\boldsymbol{\sigma}_L,\boldsymbol{\sigma}_R)\right],
\end{equation}
where $X_T(\boldsymbol{\sigma}_L,\boldsymbol{\sigma}_R) \equiv \bra{\Psi_0}\hat G(\boldsymbol{\sigma}_L)^\dagger \hat{\Pi}_T \hat G(\boldsymbol{\sigma}_R)\ket{\Psi_0}$. Given $K$ independent path pairs, the plain Monte Carlo estimator is unbiased. Crucially, each sample $X_T$ can be evaluated using the exact same randomized mixed-Pfaffian overlap primitive developed for the $W=0$ pipeline, as our phase-filtering identity rigorously extends to non-unitary propagators (Sec.~\ref{sec:nonunitary}).

\subsection{A conservative worst-case envelope for the Pfaffian estimator}
The fundamental obstruction to pushing this classical simulation to arbitrary times is not algorithmic correctness, but variance. Because the HS transformation requires complex coupling parameters ($\lambda=a+ib$ with $a=\Re\lambda$), the one-body propagators $\hat G(\boldsymbol{\sigma})$ are inherently non-unitary.

Since $V(\boldsymbol{\sigma}_\ell)$ is diagonal, $\|V(\boldsymbol{\sigma}_\ell)\|_\mathrm{op} = \max_{j}|e^{\pm \lambda\sigma_{\ell,j}}| = e^{|a|}$. Because $U_{1/2}$ is unitary, over $n$ slices this yields the uniform bound:
\begin{equation}
    \|G(\boldsymbol{\sigma})\|_\mathrm{op}
    \le \prod_{\ell=1}^{n}\|U_{1/2}V(\boldsymbol{\sigma}_\ell)U_{1/2}\|_\mathrm{op}
    \le e^{|a|n}.
    \label{eq:G_norm_bound}
\end{equation}

Our Pfaffian overlap evaluation expresses $X_T(\boldsymbol{\sigma}_L,\boldsymbol{\sigma}_R)$ via a single-particle matrix $Q_T(\boldsymbol{\sigma}_L,\boldsymbol{\sigma}_R) \equiv G(\boldsymbol{\sigma}_L)^\dagger D_T G(\boldsymbol{\sigma}_R)$. Since $D_T$ is diagonal with $\|D_T\|_\mathrm{op}=1$, we have:
\begin{equation}
    \|Q_T(\boldsymbol{\sigma}_L,\boldsymbol{\sigma}_R)\|_\mathrm{op}
    \le \|G(\boldsymbol{\sigma}_L)\|_\mathrm{op} \|G(\boldsymbol{\sigma}_R)\|_\mathrm{op}
    \le e^{2|a|n}.
    \label{eq:Q_norm_bound}
\end{equation}

For the $2r\times 2r$ skew-symmetric matrices compressed from $Q_T$, the Pfaffian is bounded by $|\pf(A)| \le \|A\|_\mathrm{op}^{r}$. This yields a conservative, pathwise worst-case envelope for the overlap primitive:
\begin{equation}
    |X_T(\boldsymbol{\sigma}_L,\boldsymbol{\sigma}_R)|
    \le
    B_{\mathrm{worst}}(W;t,\Delta t)
    \equiv
    C_T\exp \big(2|a(W,\Delta t)| n r\big),
    \label{eq:B_worst_def}
\end{equation}
where $r$ is the Pfaffian dimension parameter and $C_T$ is a constant determined by the fixed input state. 

To quantify this practical simulation boundary, we compare this rigorous envelope with a statistical random-walk approximation that provides a typical-scale proxy:
\begin{equation}
    B_{\mathrm{typ}}(W;t,\Delta t)
    \equiv \exp \left(c |a(W,\Delta t)| L\sqrt{t/\Delta t}\right),~~~ c=\sqrt{2/\pi},
\end{equation}
which captures the observed onset of numerical instability in practice.

\begin{figure}[t!]
\includegraphics[width=250px]{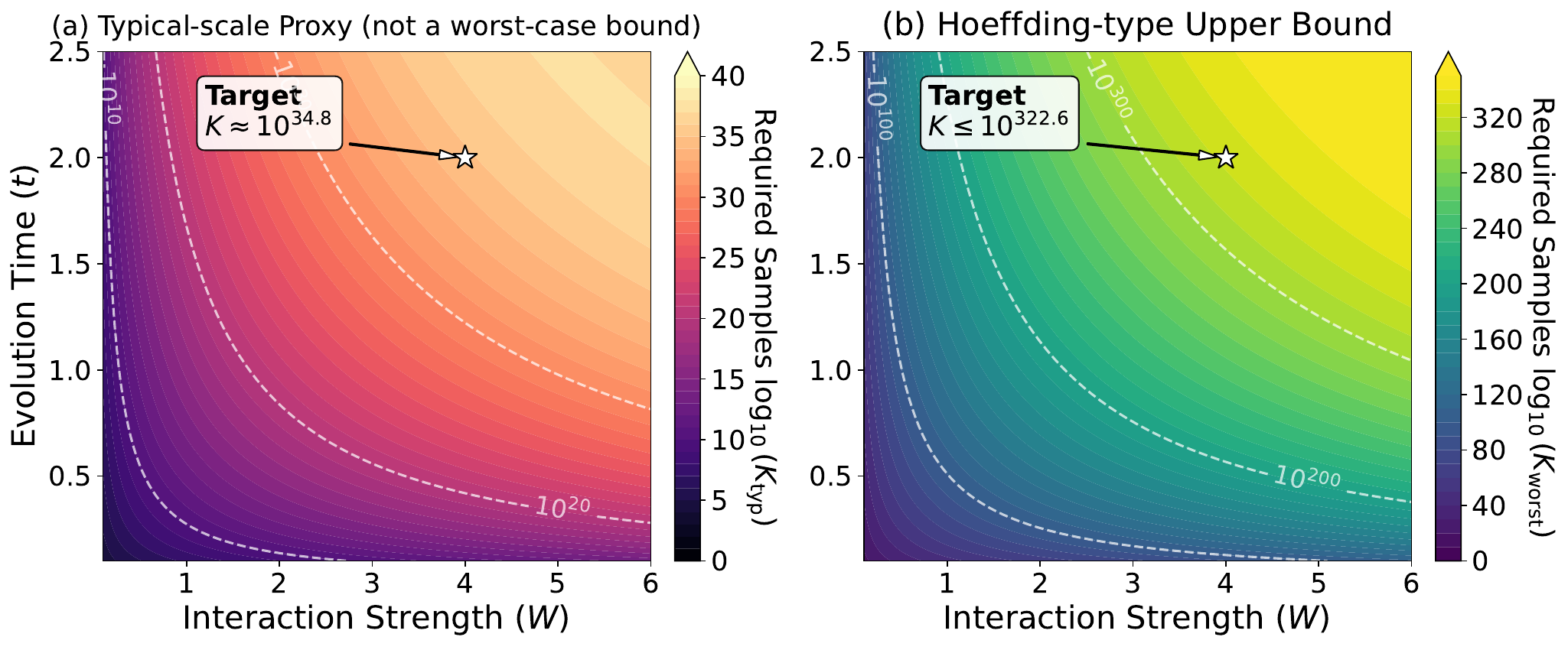} 
\caption{
Sample-complexity contours for real-time Hirsch HS sampling under Phasecraft-matched Strang Trotterization ($k=4$, $\Delta t=t/k$).
(a) Typical-scale proxy.
(b) Conservative Hoeffding-type upper bound.
Stars mark $(W,t)=(4,2)$; here $\epsilon=0.01$, $\delta=0.01$.}
\label{fig:complexity}
\end{figure}

As shown in Fig.~\ref{fig:complexity}, we convert these magnitude scales into a Hoeffding-type sample upper bound $K(B)=\frac{2B^2}{\epsilon^2}\log\frac{2}{\delta}$. We report both $K(B_{\mathrm{typ}})$ and $K(B_{\mathrm{worst}})$ as complementary indicators of practical and conservative feasibility, explicitly isolating the exact mathematical origin of the variance barrier.

\subsection{Extent-based simulation of interacting dynamics}

\begin{figure}[t!]
\includegraphics[width=250px]{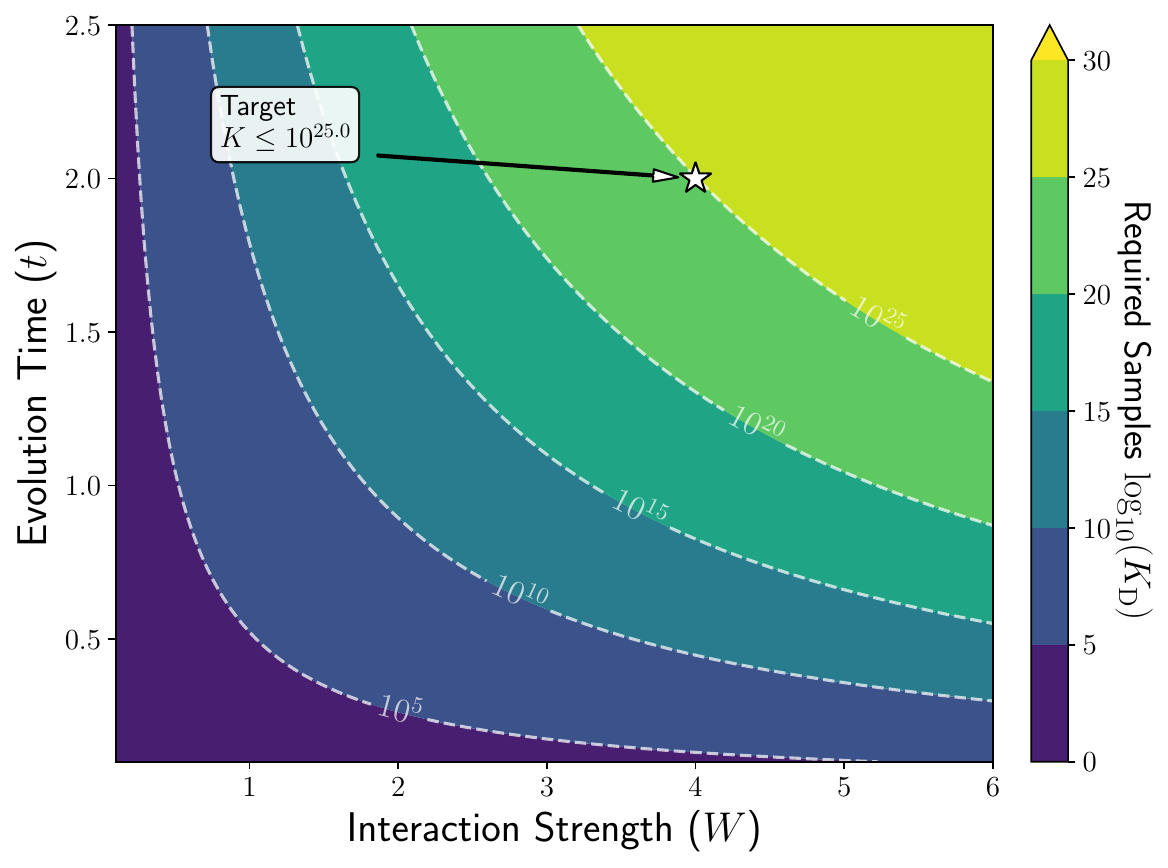} 
\caption{
Sample-complexity contours for extent-based sampling under Phasecraft-matched Trotterization ($k=4$, $\Delta t=t/k$). The star marks $(W,t)=(4,2)$; here $\epsilon=0.01$, $\delta=0.01$.}
\label{fig:extent-complexity}
\end{figure}

An alternative approach to simulating interacting dynamics is to generalize the extent-based simulation algorithms by combining them with our mixed-Pfaffian estimators. 
Assume we are interested in computing a quantity of the form 
\begin{align}
    \bra{\Phi} \hat{U}\ket{\Psi},
\end{align}
where $\hat{U}$ is expressed as a quantum circuit composed of passive fermionic-Gaussian unitaries and non-Gaussian controlled-phase gates of the form
\begin{align}
    P_{ij}(\theta) &= \exp\left(i \theta Z_i Z_j\right).
\end{align} 
The extent-optimal decomposition of Refs.~\cite{cudby2024gaussiandecompositionmagicstates,Dias_2024,ReardonSmith2024improvedsimulation} takes the form
\begin{align}
   P_{ij}(\theta)&= \cos(\theta) I + i \sin(\theta)\mathrm{diag}(1,-1,-1,1),
\end{align}
where the $2$ mode gates act on modes $i$ and $j$. Assuming there are $K$ non-Gaussian gates in the circuit we have
\begin{align}
    &\bra{\Phi} \hat{U}\ket{\Psi} =\\ &\sum_{x\in\{0,1\}^K}i^{|x|}\prod_{j=1}^K (\cos(\theta_j)^{1-x_j}\sin(\theta_j)^{x_j} ) \bra{\Phi}U(x) \ket{\Psi},
\end{align}
where each $U(x)$ is a passive fermionic-Gaussian unitary. We now rewrite this as an expectation value in terms of the probability distribution
\begin{align}
    P(x) &= \prod_{j=1}^K \frac{\cos(\theta_j)^{1-x_j}\sin(\theta_j)^{x_j}}{\cos(\theta_j) + \sin(\theta_j)}\\
    \bra{\Phi}\hat{U}\ket{\Psi} &= \sqrt{\xi}\sum_{x\in\{0,1\}^K} P(x) i^{|x|}\bra{\Phi}U(x) \ket{\Psi},
\end{align}
where we have introduced the extent
\begin{align}
    \xi = \prod_j(\cos(\theta_j)^{1-x_j}\sin(\theta_j)^{x_j} )^2.
\end{align}
We now apply the same reasoning as Theorem~\ref{thm:overlap_apsg} to express this as the expectation value of a random variable which can be computed in polynomial time, and is bounded point-wise by $\sqrt{\xi}$.

In Fig.~\ref{fig:extent-complexity}, we show the sample complexity of applying this scheme to estimate expectation values of observables relevant to the experiment of Ref.~\cite{alam2025fermionic}.

\section{One-body sources and derivative identities}
\label{app:chem_sources}

This Appendix presents source identities that reduce transition one- and two-body moments to local derivatives of the same number-preserving Gaussian kernel used throughout the paper.

For number-preserving Gaussian maps \(\hat G_L\) and \(\hat G_R\), define
\begin{align}
    \mathcal{K}_{\Phi,\Psi}(J;\hat G_L,\hat G_R)
    &\equiv
    \langle \Phi|
    \hat G_L
    e^{\hat c^\dagger J \hat c}
    \hat G_R
    |\Psi\rangle, \\
    \hat c^\dagger J \hat c &\equiv \sum_{p,q}J_{pq}\hat c_p^\dagger\hat c_q .
\end{align}
The matrix \(J\in\mathbb C^{M\times M}\) is treated as a formal complex source.
The corresponding single-particle matrix of the middle Gaussian factor is \(e^J\), so the total single-particle map entering the Pfaffian kernel is \(G_L e^J G_R\).

The transition one-body moments are
\begin{equation}
\langle \Phi|
\hat G_L \hat c_p^\dagger \hat c_q \hat G_R
|\Psi\rangle
=
\left.
\frac{\partial}{\partial J_{pq}}
\mathcal{K}_{\Phi,\Psi}(J;\hat G_L,\hat G_R)
\right|_{J=0}.
\label{eq:app:1rdm_source}
\end{equation}

For transition two-body RDMs, it is cleaner to use ordered one-body sources.
Let
\begin{align}
    \hat O_{ab}\equiv \hat c_a^\dagger \hat c_b .
\end{align}
Then
\begin{align}
&\left.
\frac{\partial^2}{\partial s\,\partial t}
\langle \Phi|
\hat G_L
e^{s\hat O_{pr}}
e^{t\hat O_{qs}}
\hat G_R
|\Psi\rangle
\right|_{s=t=0}
\nonumber\\
&\qquad =
\langle\Phi|
\hat G_L
\hat O_{pr}\hat O_{qs}
\hat G_R
|\Psi\rangle .
\end{align}
Using
\begin{align}
    \hat O_{pr}\hat O_{qs}
    =
    \delta_{rq}\hat O_{ps}
    +
    \hat c_p^\dagger \hat c_q^\dagger \hat c_s \hat c_r ,
\end{align}
we obtain
\begin{align}
    &\langle\Phi|
    \hat G_L
    \hat c_p^\dagger \hat c_q^\dagger \hat c_s \hat c_r
    \hat G_R
    |\Psi\rangle\\
    &=
    \left.
    \frac{\partial^2}{\partial s\,\partial t}
    \langle \Phi|
    \hat G_L
    e^{s\hat O_{pr}}
    e^{t\hat O_{qs}}
    \hat G_R
    |\Psi\rangle
    \right|_{s=t=0} \\ 
    &-
    \delta_{rq}
    \langle\Phi|
    \hat G_L
    \hat c_p^\dagger \hat c_s
    \hat G_R
    |\Psi\rangle.
    \label{eq:app:2rdm_ordered_source}
\end{align}

For a fixed one-body operator
\begin{align}
    \hat O(K)=\hat c^\dagger K\hat c,
\end{align}
define
\begin{equation}
    \chi_{\Phi,\Psi}^{(K)}(\theta;\hat G_L,\hat G_R)
    \equiv
    \langle \Phi|
    \hat G_L e^{i\theta \hat O(K)}\hat G_R
    |\Psi\rangle .
\end{equation}
Then
\begin{align}
    \langle \Phi|
    \hat G_L \hat O(K)\hat G_R
    |\Psi\rangle
    &=
    \left.
    \frac{1}{i}\frac{d}{d\theta}
    \chi_{\Phi,\Psi}^{(K)}(\theta;\hat G_L,\hat G_R)
    \right|_{\theta=0},\\
    \langle \Phi|
    \hat G_L \hat O(K)^2\hat G_R
    |\Psi\rangle
    &=
    -
    \left.
    \frac{d^2}{d\theta^2}
    \chi_{\Phi,\Psi}^{(K)}(\theta;\hat G_L,\hat G_R)
    \right|_{\theta=0}.
    \label{eq:app:O2_source}
\end{align}
Combining these identities with the molecular sum-of-squares factorization in
Eq.~\eqref{eq:chem:sum_of_squares} gives a systematic route to transition
Hamiltonian matrix elements and orbital-gradient terms.

\end{document}